\documentclass[11pt, reqno]{article}
\usepackage[left=24mm,top=2cm,right=24mm,bottom=3cm, footskip=1cm, nohead]{geometry}
\usepackage[utf8]{inputenc}
\usepackage{graphicx}
\usepackage[usenames,dvipsnames]{color}
\usepackage{amsmath}
\usepackage{amssymb}
\usepackage{amsthm}  
\usepackage[colorlinks=true,linkcolor=black,citecolor=black,plainpages=false,pdfpagelabels]{hyperref}
\usepackage{color}
\usepackage{times}
\usepackage{float}

\newtheorem{thm}{Theorem}[section]
\newtheorem{cor}[thm]{Corollary}
\newtheorem{lem}[thm]{Lemma}
\newtheorem{defin}{Definition}[section]

\numberwithin{equation}{section}

\newcommand{\mcal}{\mathcal}

\newcommand{\bra}[1]{\langle #1|}
\newcommand{\ket}[1]{|#1 \rangle}
\newcommand{\braket}[2]{\langle #1|#2\rangle}
\newcommand{\ketbra}[1]{\ket{#1}\bra{#1}}
\newcommand{\ketb}[2]{\ket{#1}\bra{#2}}
\newcommand{\ident}{\mathbb{I}}
\DeclareMathOperator{\tr}{Tr}
\DeclareMathOperator{\rank}{rank}
\newcommand{\mdag}{^{\dag}} 
\newcommand{\demi}{\frac{1}{2}}

\def\ox{\otimes}
\def\>{\rangle}
\def\<{\langle}

\def\Tr{\mathrm{Tr}}
\def\Pr{\mathrm{Pr}}

\def\lp{\left(}
\def\rp{\right)}
\def\ls{\left[}
\def\rs{\right]}
\def\lb{\left\{}
\def\rb{\right\}}

\DeclareMathOperator{\op}{op}

\DeclareMathOperator{\Herm}{Herm}
\DeclareMathOperator{\Pos}{Pos}

\DeclareMathOperator{\acc}{acc}
\DeclareMathOperator{\polylog}{polylog}

\newcommand{\mbE}{\mathbb{E}}

\newcommand{\mbR}{\mathbb{R}}
\newcommand{\mbS}{\mathbb{S}}

\newcommand{\mfX}{\mathfrak{X}}
\newcommand{\mfY}{\mathfrak{Y}}

\newcommand{\sfA}{\mathsf{A}}

\newcommand{\expU}{\underset{U}{\mbE}}
\newcommand{\PrU}{\underset{U}{\Pr}}

\newcommand{\lno}{\left\|}
\newcommand{\rno}{\right\|}

\begin{document}

\title{Locking classical information}

\author{
	Fr\'ed\'eric Dupuis\thanks{Institute for Theoretical Physics, ETH Zurich, Switzerland} \and
	Jan Florjanczyk\thanks{School of Computer Science, McGill University, Montreal, Canada} \and
	Patrick Hayden$^\dagger$\thanks{Perimeter Institute for Theoretical Physics, Waterloo, Canada} \and
	Debbie Leung\thanks{Institute for Quantum Computing, University of Waterloo, Waterloo, Canada}\,\,$^\ddagger$
}

\date{4 November 2010}

\pagestyle{plain}

\maketitle

\begin{abstract}
It is known that the maximum classical mutual information that can be achieved between measurements on a pair of quantum systems can drastically underestimate the quantum mutual information between those systems. In this article, we quantify this distinction between classical and quantum information by demonstrating that after removing a logarithmic-sized quantum system from one half of a pair of perfectly correlated bitstrings, even the most sensitive pair of measurements might only yield outcomes essentially independent of each other. This effect is a form of information locking but the definition we use is strictly stronger than those used previously. Moreover, we find that this property is generic, in the sense that it occurs when removing a random subsystem. As such, the effect might be relevant to statistical mechanics or black hole physics. Previous work on information locking had always assumed a uniform message. In this article, we assume only a min-entropy bound on the message and also explore the effect of entanglement. We find that classical information is strongly locked almost until it can be completely decoded. As a cryptographic application of these results, we exhibit a quantum key distribution protocol that is ``secure'' if the eavesdropper's information about the secret key is measured using the accessible information but in which leakage of even a logarithmic number of key bits compromises the secrecy of all the others.

\end{abstract}

\section{Introduction}
One of the most basic and intuitive properties of most information measures is that the amount of information carried by a physical system must be bounded by its size. For example, if one receives ten physical bits, then one's information, regardless of what that information is ``about'', should not increase by more than ten bits. While this is true for most information measures, in quantum mechanics there exist natural ways of measuring information that violate this principle by a wide margin. In particular, this violation occurs when one defines the information contained in a quantum system as the amount of classical information that can be extracted by the best possible measurement. To construct examples of this effect, we take a classical message and encode it into a two-part quantum message: a \emph{cyphertext}, which is roughly as large as the message, and a much smaller \emph{key}. Given both the cyphertext and the key, the message can be perfectly retrieved. We can then look at the amount of information that can be extracted about the message by a measurement given only access to the cyphertext. Locking occurs if this amount of information is less than the amount of information in the message minus the size of the key.

In previous work on locking \cite{dhlst03,hlsw03}, this amount of information was taken to be the accessible information, the maximum (classical) mutual information between the message and the result of a measurement. In \cite{dhlst03}, the authors constructed the first example of locking as follows: the cyphertext consists of the uniformly random message, encoded in one of two mutually unbiased bases, and the (one-bit) key reveals the basis in which the encoding was done. In this example, given only the cyphertext, the classical mutual information is only $\frac{n}{2}$ for an $n$-bit message. Hence, the one-bit key can increase the classical mutual information by another $\frac{n}{2}$ bits. In \cite{hlsw03}, the authors considered a protocol in which one encodes a classical message using a fixed basis, and then applies one of $k$ fixed unitaries (where $k = O(\polylog n + \log \frac{1}{\varepsilon})$); the classical key reveals which unitary was applied. If the unitaries are chosen according to the Haar measure, then with high probability, the accessible information was shown to be at most $\varepsilon n$ when one only has the cyphertext.

In this paper, we present stronger and more general locking results, and show that this effect is generic. Our results will be stronger in the sense that instead of using the accessible information, we will define locking in terms of the trace distance between measurement results on the real state and measurement results on a state completely independent of the message (see Definition \ref{def:locking}). Unlike the accessible information, this has a very natural operational interpretation: it bounds the largest probability with which we can guess, given a message $m$ and the result $x$ of a measurement done on a cyphertext, whether $x$ comes from a valid cyphertext for $m$ or from a cyphertext generated independently of $m$. In other words, one could almost perfectly reproduce any measurement results made on a valid cyphertext without having access to the cyphertext at all. Moreover, we recover a strengthened version the earlier statements about the accessible information. Whereas previously the accessible information was shown to be at most 3 bits, our techniques show that the accessible information can be made arbitrarily small. (A follow-up paper further strengthens the definition and explores connections to low-distortion embeddings~\cite{FawHaySen10}.) 

Despite this stronger definition, we will be able to show that the locking phenomenon is generic. Instead of having a classical key reveal the basis in which the information is encoded, as in \cite{dhlst03,hlsw03}, we consider the case where there is a single unitary, and the key is simply a small part of the quantum system after the unitary is applied. This means that we can make not only cryptographic statements, but also statements about the dynamics of physical systems, where the unitary represents the evolution of the system. In particular, we will be able to show that locking occurs with high probability in physical systems whose internal dynamics are sufficiently generic to be adequately modelled by a Haar-distributed unitary. This can therefore give interesting results in the context of thermodynamics, or of the black hole information problem.

In that vein, we will also allow the measuring device to share entanglement with the cyphertext-key compound system. While this may not correspond to a very meaningful cryptographic scenario, it allows us to study the behavior of entanglement in physical systems, and to study the extent to which the presence of entanglement interferes with this locking effect.

Finally, in contrast to previous studies, we will not limit the message (or the entanglement) to be uniform; the size of the key will instead depend on the min-entropy of the message. This assumption is easier to justify in cryptographic applications. Indeed, while the locking results we present here can be interpreted as demonstrating the possibility of encrypting classical messages in quantum systems using only very small keys, care must be taking when composing such encryption with other protocols. We use our results to exhibit a quantum key distribution protocol, for example, that appears to be secure if the eavesdropper's information about the secret key is measured using the accessible information, but in which leakage of a logarithmic amount of key causes the entire key to be compromised.



\subsection{Transmitting information through a generic unitary} \label{sec:how-to-lock}

To end the introduction, we introduce
the physical scenario that will occupy us throughout the article. The situation is depicted in Figure~\ref{fig:scenario}. 

\begin{figure}[h!]
	\begin{center}
	\includegraphics[width=0.65\linewidth]{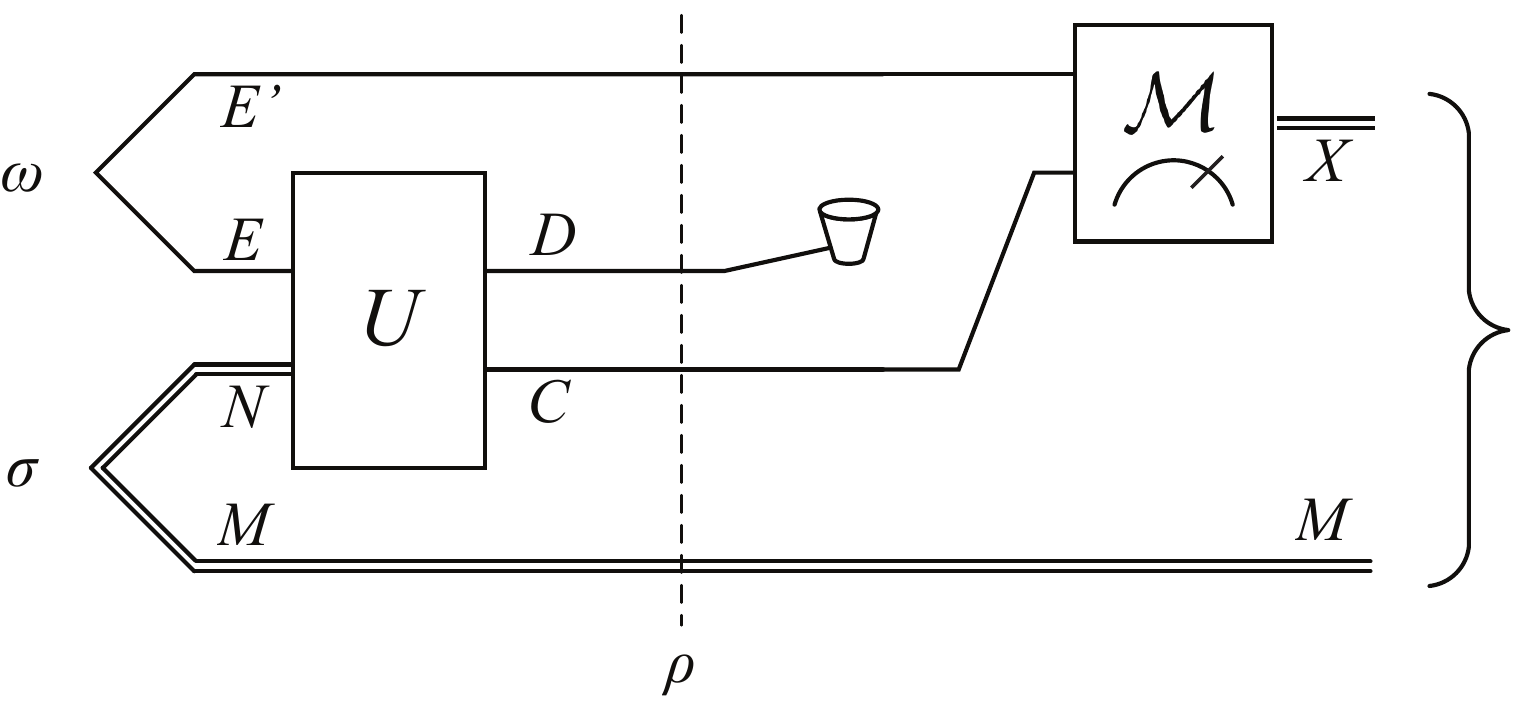}
	\end{center}
	\caption{A quantum circuit depicting the physical scenario. The classical message $M$ gets encoded in $N$, and the unitary then mixes it with the $E$ part of the shared entanglement. If the information is locked, any joint measurement $\mathcal{M}$ on $C$ and $E'$ will yield a result $X$ that is almost independent of the message. On the other hand, if $C$ is large enough, there will be a joint measurement $\mathcal{M}$ reliably decoding $M$.}
\label{fig:scenario}
\end{figure}

Now, let $\{ \ket{\psi_m} : 1 \leqslant m\leqslant |M| \}$ be any orthonormal basis for $N$. The analysis will focus on the properties of the states 
\begin{eqnarray}
	\sigma^{MN} & := & \sum_{m=1}^{|M|} p_m \ketbra{m}^M \otimes \ketbra{\psi_m}^{N}  
		\quad \mbox{and} \label{eqn:sigmadef} \\
	\rho^{MCDE'} & := & \lp \ident^{ME'} \ox U^{NE\rightarrow CD} \rp \lp \sigma^{MN} \ox \omega^{EE'} \rp
		\lp \ident^{ME'} \ox U^{NE\rightarrow CD} \rp^\dagger. \label{eqn:rhodef} 
\end{eqnarray}

Our objective is to demonstrate that until $C$ is large enough that there exists a measurement on $CE'$ capable of revealing \emph{all} the information about the message $M$, no measurement will reveal \emph{any} information about the message. This can't quite be true, of course, so what we will demonstrate is that the jump from no information to complete information involves enlarging $C$ by a number of qubits logarithmic in the size of the message $M$ and the amount of entanglement $E$.

Assume for simplicity both that $M$ is uniformly distributed and that the state $\omega^{EE'}$ is maximally entangled. As a first step, it is necessary to determine how large $C$ needs to be in order for there to exist a measurement on $CE'$ that will reveal the message $M$. Begin by purifying the state $\sigma$ to
\begin{equation}
|\sigma\rangle^{RMN} = \frac{1}{\sqrt{|M|}} \sum_{m=1}^{|M|}
	| m \rangle^{R} \otimes |m\rangle^{M} \otimes |\psi_m\rangle^{N}.
\end{equation}
Even more demanding than performing a measurement to reveal $m$ is the task of transmitting the quantum information about $RM$ through $U$, allowing the decoder, who has access only to $CE'$, to recover a high fidelity copy of the state $|\sigma\rangle^{RMN}$. If $U$ is selected according to the Haar measure, then Theorem IV.1 of \cite{FQSW} implies that there is a quantum operation $\mathcal{D}^{CE' \rightarrow N}$ acting only on $CE'$ such that
\begin{equation} \label{eqn:purified-decoding}
	\left\| \mathcal{D} \left( \tr_{D} \left[ U^{NE\rightarrow CD}( \sigma^{RMN} \otimes \omega^{EE'}) (U^{NE\rightarrow CD})^\dagger \right] \right) 
	- \sigma^{RMN} \right\|_1 
	\leq 2 \sqrt{ \frac{M}{C} }.
\end{equation}
Because the trace distance is monotonic under quantum operations, it will not increase by taking the partial trace over $R$ and measuring in the basis $\{ | \psi_m \rangle \}$~\cite{NC2000}. If we let $p(m' | m)$ be the probability of getting an outcome $|\psi_{m'}\rangle$ when the message was in fact $m$, Equation (\ref{eqn:purified-decoding}) therefore implies that
\begin{equation} \label{eqn:measurement-existence}
\frac{1}{M} \sum_m \sum_{m' \neq m} p(m' | m ) \leq \sqrt{ \frac{M}{C} }.
\end{equation}
In words, the probability of the measurement yielding the incorrect outcome, averaged over all messages, is at most $\sqrt{M/C}$, so as soon as $C$ is significantly larger than $M$, a measurement on $CE'$ can be found that will reveal the message. Our goal in this article will be to demonstrate that until this condition is met, no measurement will reveal any significant information about the message.


\subsection{Structure of the paper}
The next subsection explains the notation used throughout the paper, and we then move on to the formal definition of locking as well as other important concepts in Section \ref{sec:definitions}. Section \ref{sec:main-results} will state the main results and give a high-level overview of the proof, and Section \ref{sec:preliminary-calculations} will begin the proof with some key lemmas. Section \ref{sec:projective-measurements} will deal with the proofs of our theorems in the easier case where the measurement device is restricted to making only projective measurements, and Section \ref{sec:povms} will deal with the case of general measurements (POVMs). We then show in Section \ref{sec:packing} that, in many regimes, as soon as the information is not locked, it is completely decodable. Implications for the security definitions of quantum cryptographic protocols will be presented in Section \ref{sec:qkd}, and we conclude the paper with a discussion in Section \ref{sec:discussion}.

\subsection{Notation}
\label{notation-at-beginning-firstpage}

\newcommand{\tabstart}[1]{\noindent \begin{center}\begin{tabular}{p{2.95cm}p{10cm}}
    \multicolumn{2}{l}{{\bf #1}} \\ \hline \\[-2.5ex] } 

\newcommand{\tabstop}{\\ \hline \end{tabular}\end{center}}

\newcommand{\tabinter}{\vspace{2ex}}

\tabstart{General}
$\log$ & Logarithm base 2.\\
$\mbE_U[f(U)]$ & Expectation value of $f(U)$ over the random variable $U$. \\
$AB$ & Composite quantum system whose associated Hilbert space is $A \otimes B$. We frequently identify quantum systems with their associated Hilbert spaces.\\
$|A|$ & Dimension of Hilbert space $A$. However, we will often drop the $|\cdot|$. For example, the dimension of the composite system $MCK$ is denoted by $MCK$ (a scalar value).\\
$A^{\ox 2}$ & Two identical copies of $A$ the second of which is denoted by $\overline{A}$. \\
$\ket{\psi}^A, \ket{\varphi}^A, \dots$ & Vectors in $A$.\\
$\psi^A, \varphi^A, \dots$ & The ``unketted'' versions denote their associated density matrices: $\psi^A = \ketbra{\psi}$. Furthermore, if we have defined a state $\psi^{AB}$, then $\psi^A = \tr_B[\psi^{AB}]$.\\
$\pi^A$ & The maximally mixed state $\frac{\ident^A}{|A|}$ . \\
$\mcal{U}(A)$ & The unitary group on $A$. \\
$\Pos(A)$ & The subset of Hermitian operators from $A$ to $A$ consisting of positive semidefinite matrices.\\
$\mcal{L}(s,\eta)$ & The set of all $(s,\eta)$-quasi-measurements, see Definition \ref{def:quasi-measurement}.
\tabstop

\tabinter

\tabstart{Operators}
$\ident^A$ & Identity operator on $A$.\\
$M^{A \rightarrow B}$ & Indicates that the operator $M$ is a transformation from states on $A$ to states on $B$. \\
$\mcal{M}^{A \rightarrow B}$ & Indicates that the superoperator $\mcal{M}$ is a transformation from operators on $A$ to operators on $B$. $M$ and $\mcal{M}$ will be freely identified with their extensions (via tensor product with the identity) to larger systems.\\
$M \cdot N$ & $MNM\mdag$\\
$M \leqslant N$ & If $M, N \in \Herm(\sfA)$, this means that $N - M \in \Pos(A)$. \\
$\sqrt{M}$ & If $M \in \Pos(A)$ has spectral decomposition $M = \sum_i \lambda_i \ketbra{\psi_i}$, then $\sqrt{M} = \sum_i \sqrt{\lambda_i} \ketbra{\psi_i}$. \\
$\Pi^{A}_{\pm}$ & Projector onto the symmetric ($+$) or antisymmetric ($-$) subspace of $A^{\ox2}$. \\
$\op_{A \rightarrow B}(\ket{\psi}^{AB})$ & Turns a vector into an operator. See Definition \ref{def:opAB}.
\tabstop

\tabinter

\tabstart{Norms and Entropies}
$\left\| M^{A \rightarrow B} \right\|_1$ & $\tr\sqrt{M\mdag M}$\\
$\left\| \ket{\psi} \right\|_2$ & $\sqrt{|\braket{\psi}{\psi}|}$\\
$\left\| M^{A \rightarrow B} \right\|_2$ & $\sqrt{\tr[M \mdag M]}$\\
$\left\| M^{A \rightarrow B} \right\|_{\infty}$ & Largest singular value of $M$, \emph{i.e.} the operator norm of $M$.\\
$H_2(A)_{\rho}$ & Renyi 2-entropy of $A$, defined as $-\log \Tr [\rho^2]$. \\
$H_{\min}(A)_{\rho}$ & Quantum min-entropy of $A$, defined as $-\log \min_{\lambda \in \mbR} \{ \lambda : \rho^A \leqslant \lambda \ident^A \}$.\\
$H_{\max}(A)_{\rho}$ & Quantum max-entropy of $A$, defined as $2 \log \tr \sqrt{\rho^A}$.\\
$I(A; B)_{\rho}$ & Mutual information of $A$ and $B$, defined as $H(A)_{\rho}+ H(B)_{\rho} - H(AB)_{\rho}$. \\
$I_{\mathrm{acc}}(A; B)_{\rho}$ & Accessible information, see Definition \ref{def:accessibleinformation}. 
\tabstop
\tabinter

\section{Definitions}\label{sec:definitions}


This section will present the basic definitions needed to state our results. First, it will be very convenient for us to represent measurements via superoperators in the following manner:
\begin{defin}[Measurement superoperator]
	We call a completely positive, trace-preserving (CPTP) map $\mathcal{M} : \mathcal{B}(A) \rightarrow \mathcal{B}(X)$ a \emph{measurement superoperator} if it is of the form $\mathcal{M}(\rho) = \sum_{i=1}^N \ketbra{i}^X \tr[M_i^A \rho]$, where $\{ \ket{i}^A : i \in \{1,\dots,N\} \}$ is an orthonormal basis for $X$, each $M_i^A$ is positive semidefinite, and $\sum_{i=1}^N M_i^A = \ident^A$.
\end{defin}

These play a central role in the definition of accessible information.
\begin{defin}[Accessible information~\cite{fuchs-phd}]\label{def:accessibleinformation}
Let $\rho^{AB}$ be a quantum state. Then, the accessible information $I_{\acc}(A;B)$ is defined as
\[ I_{\acc}(A;B)_{\rho} := \sup_{\mathcal{A}, \mathcal{B}} I(X;Y)_{(\mathcal{A} \otimes \mathcal{B})(\rho)}, \]
where $\mathcal{A}^{A \rightarrow X}$ and $\mathcal{B}^{B \rightarrow Y}$ are measurement superoperators, and the supremum is taken over all possible superoperators. In other words, the accessible information is the largest possible mutual information between the results of measurements made on $A$ and $B$.
\end{defin}
The accessible information was originally defined only for states in which the $A$ subsystem was classical. In that case, no measurement on $A$ is necessary in the optimization. This quantity is also known as the classical mutual information of a quantum state~\cite{OlZur01,HenVed01}.

We also need to introduce the concept of \emph{quasi-measurements} for our analysis. They are, as their name indicates, almost measurements, but differ in three ways: they only contain rank-one elements of equal weight, they have exactly $n$ outcomes, and the sum of all the elements does not necessarily equal the identity, but is instead bounded by $k\ident$:
\begin{defin}[Quasi-measurement]\label{def:quasi-measurement}
	We call a superoperator $\mathcal{M}^{A \rightarrow B}$ an \emph{$(s,\eta)$-quasi-measurement} if it is of the form $\mathcal{M}(\rho) = \frac{|A|}{s} \sum_{i=1}^s \ketb{i}{\chi_i} \rho \ketb{\chi_i}{i}$ where the $\ket{i}$ index an orthonormal basis for $B$, and $\frac{|A|}{s}\sum_{i=1}^s  \ketbra{\chi_i} \leqslant \eta \ident^{A}$.  We call the set of all $(s,\eta)$-quasi-measurements on a given system, $\mcal{L}(s,\eta)$.
\end{defin}
The reason for introducing these, as will soon become apparent, is that they are almost equivalent to POVMs for our purposes while being much easier to handle mathematically. It can easily be seen that projective measurements are simply $(A,1)$-quasi-measurements.

We now give the formal, strengthened definition of locking. The states in question were introduced in Section \ref{sec:how-to-lock}. However, because the cyphertext will always be smaller than or equal to the message when locking occurs, certain identifications become possible. In particular, we can assume without loss of generality that $N \cong C \otimes K$ and $D \cong E \otimes K$. Since the analysis will be performed using only $C$, $K$ and $E$, we reproduce the illustration of the physical scenario with the identifications made in Figure~\ref{fig:locking}.
\begin{figure}[h!]
	\begin{center}
	\includegraphics[width=0.65\linewidth]{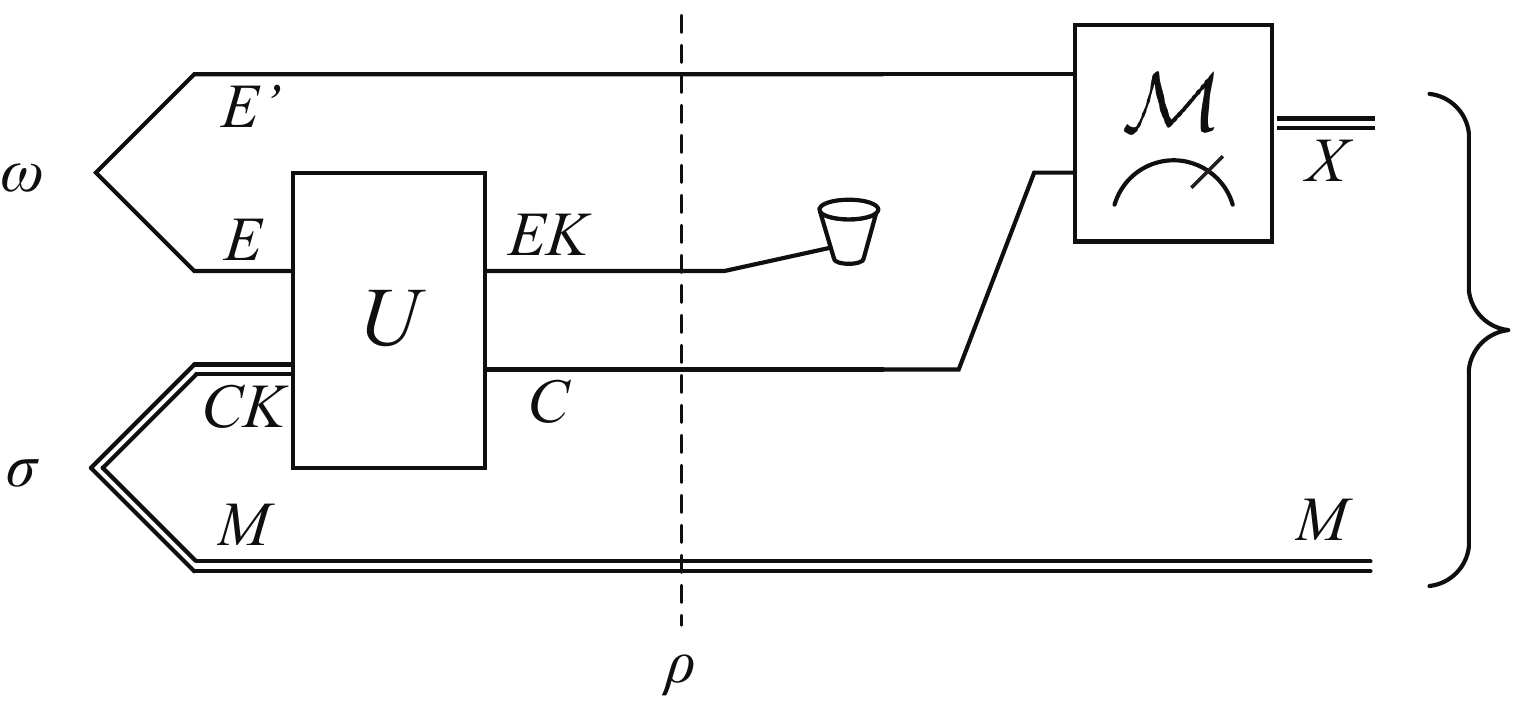}
	\end{center}
	\caption{A quantum circuit depicting the physical scenario with the locking-specific identifications $N \cong C \otimes K$ and $D \cong E \otimes K$ made.}
\label{fig:locking}
\end{figure}
\begin{defin}[$\varepsilon$-locking scheme]\label{def:locking}
	Let $M, C, K, E$ and $E'$ be quantum systems. Let $\rho^{MCKEE'}$ be a quantum state of the form
	\begin{equation}\label{eqn:def-rho}
	\rho^{MCKEE'} = \sum_m p_m U^{CKE} \left(\ketbra{m}^M \otimes \ketbra{\psi_m}^{CK} \otimes \ketbra{\omega}^{EE'} \right) {U^{CKE}}^{\dagger},
\end{equation}
	where the $|\psi_m\rangle$ are orthogonal and $U^{CKE}$ is unitary. Then we call $\rho$ an $\varepsilon$-locking scheme if for any measurement superoperator $\mathcal{M}^{CE' \rightarrow X}$, we have that
	\[ \left\| \mathcal{M}\left( \rho^{MCE'} \right) - \mathcal{M}\left( \rho^M \otimes \rho^{CE'} \right) \right\|_1 \leqslant \varepsilon. \]
\end{defin}

Note that this definition of locking is rather different from that used in previous work in the area (\cite{dhlst03,hlsw03}). Their definition involved the \emph{accessible information} between the cyphertext and the message. We can show that our definition implies the older one:
\begin{lem} \label{lem:alicki-fannes}
	Let $\xi^{MB}$ be a cq-state such that $\left\| \mathcal{M}(\xi^{MB}) - \xi^M \otimes \mathcal{M}(\xi^B) \right\|_1 \leqslant \varepsilon$ for all measurement superoperators $\mathcal{M}^{B \rightarrow X}$. Then,
\[I_{\acc}(M;B)_{\xi} \leqslant 4\varepsilon \log M + 2\eta (1-\varepsilon) + 2\eta (\varepsilon),\]
where $\eta(x) := -x \log x$ and $\eta(0) = 0$.
\end{lem}
\begin{proof}
	This is a direct application of the Alicki-Fannes inequality \cite{alicki-fannes}.
\end{proof}


Four quantities will be particularly useful for quantifying variations from uniform messages and maximal entanglement,
\begin{eqnarray}
	\Delta_{M, \infty} & := & 2^{\log M - H_{\mathrm{min}}(M)_{\sigma}}, \label{eqn:deltaMmindef} \\
	\Delta_{M, 2} & := & 2^{\log M - H_{2}(M)_{\sigma}}, \label{eqn:deltaM2def} \\
	\Delta_{E, \infty} & := & 2^{\log E - H_{\mathrm{min}}(E)_{\omega}}, \label{eqn:deltaEmindef} \\
	\Delta_{E, 2} & := & 2^{\log E - H_{2}(E)_{\omega}}. \label{eqn:deltaE2def}
\end{eqnarray}
For a point mass distribution $p_m$, $\Delta_{M, \infty} = \Delta_{M, 2} = |M|$ and for the uniform distribution $\Delta_{M, \infty} = \Delta_{M, 2} = 1$. To give an interpretation of the $\Delta_E$ quantities, we can note that for a bipartite $\ket{\omega}^{EE'}$ with no entanglement, $\Delta_{E, \infty} = \Delta_{E, 2} = |E|$. However, if $\ket{\omega}^{EE'}$ is the maximally entangled state, then $\Delta_{E, \infty} = \Delta_{E, 2} = 1$, which we call maximal entanglement. The case of a uniformly distributed message and maximal entanglement will give the simplest expressions for minimum key size. The $\Delta$ terms are used in the calculations to provide more general statements relating the entropy of the message and entanglement to the key size.


\section{Main results and proof sketch}\label{sec:main-results}
The locking scheme we study is a scheme where the unitary in Definition \ref{def:locking} is chosen according to the Haar measure. Let $c,e$ and $n$ be the logarithms of $|C|,|E|$, and $|M|=|N|$ respectively.  In particular, the message is $n$ bits long. Define $K = M/C$ and $k = \log K$. Then $k = n - c$ is the difference in size between the message and cyphertext, that is, the size of the key.  Our main theorem is the following:
\begin{thm} \label{thm:locking-intro}
	If $U$ is chosen according to the Haar measure, then the scheme described in Definition \ref{def:locking} is an $\varepsilon$-locking scheme with probability at least $1-2^{-9(|C||E|)^2}$ if
	\begin{equation*}
	 k > \demi \Big(n - H_{\min}(M)_{\sigma}\Big) + \demi \Big(e - H_{\min}(E)_{\omega}\Big)\\
	+ \log(c+e) + 2 \log(1/\varepsilon) + 11
	\end{equation*}
	as long as $\varepsilon > 16 \Delta_{E,\infty}/\sqrt{|KE|}$.
\end{thm}

For the cryptographically relevant case in which there is no entanglement shared with the measuring device, we therefore get:
\begin{cor}
	If $U$ is chosen according to the Haar measure, then the scheme described in Definition \ref{def:locking} without shared entanglement is an $\varepsilon$-locking scheme with probability at least $1-2^{-9|C|^2}$ if 
	\begin{equation*}
		k > \demi \Big(n - H_{\min}(M)_{\sigma}\Big) + \log c\\
		+ 2\log(1/\varepsilon) + 11
	\end{equation*}
	as long as $\varepsilon > 16/\sqrt{|K|}$.
\end{cor}
Hence, the size of the key must be at least as large as half the ``min-entropy deficit'' ($n - H_{\min}(M)_{\sigma}$) of the message plus a term of the order of the logarithm of the size of the message. In particular, for a uniform message, the min-entropy deficit is zero, and a key of size at least $\log c + 2\log(1/\varepsilon) + 11$ is sufficient for locking.

Conversely, we can show that in certain regimes, if the information is not locked, then it is completely decodable, with almost no middle ground. More precisely, we have the following:

\begin{thm} \label{thm:decode-intro}
	If $U$ is chosen according to the Haar measure, then the information in the scheme described in Figure~\ref{fig:scenario} without shared entanglement is asymptotically almost surely decodable to within $\varepsilon$ in trace distance for a receiver with only $C$ as long as
\begin{equation*}
	k  \leqslant \demi \Big( n - H_{\max}(M)_{\sigma} \Big) - \demi \Big( e - H_2(E)_{\omega} \Big) - 2 \log(1/\varepsilon) - 4.
\end{equation*}
\end{thm}

Note that decoding the message will often require that the cyphertext be longer than the message, in which case $k$ will be negative. Comparing Theorems \ref{thm:locking-intro} and \ref{thm:decode-intro} reveals that the difference between being $\varepsilon$-locked and being able to decode quantum information to within $\varepsilon$ is determined by at most
\begin{equation*}
\frac{1}{2} \left[ H_{\max}(M)_\sigma - H_{\min}(M)_\sigma \right]
+ \left[ e - H_{\min}(E)_\omega \right] + \log(c+e) + 4 \log(1/\varepsilon) + 15
\end{equation*}
qubits,
where the inequality $H_2 \geq H_{\min}$ has been used to simplify the expression. 

In other words, if we consider the case of maximal entanglement, then the gap between locking and decodability can only be as wide as the difference between the min- and max-entropy of the message modulo logarithmic terms. One should note that this gap is real, and not only an artifact of our proof technique. To see this, consider an $n$-bit message distributed such that with probability $\demi$, the first bit is uniform and the rest of the string is always zero, and with probability $\demi$ the whole string is uniform. The max-entropy of such a message is $n$, but the min-entropy is tiny. Now, to be able to decode, one must be able to decode the entire string in the ``worst-case'' scenario where the whole string is uniform, so the max-entropy is relevant in this case. But in the locking case, we must be able to lock in the worst-case scenario of only one bit being random, so the min-entropy is the relevant quantity here.

The effect of non-maximal entanglement is not entirely clear, however. There is a fairly large gap between our locking and decodability results here, but the locking side is almost certainly not tight in general. For instance, we can easily set up the system in such a way that there is a part of $E'$ that is clearly useless, but our proof technique forces us to take this part into account, which artificially hurts our bound. This extreme case can be ruled out by restricting $E'$ to the support of $\omega^{E'}$, but it seems likely that more gains could be found in the general case.

Finally, in addition to studying locking for its own sake, we use our results to exhibit a quantum key distribution protocol that appears to be secure if the eavesdropper's information about the secret key is measured using the accessible information, but in which leakage of a logarithmic amount of key causes the entire key to be compromised. This is done in Section \ref{sec:qkd}.

\subsection{Proof sketch}
We will give here a very high-level overview of the proof. The basic idea is to start from the fact that, given a fixed measurement superoperator, the probability over the choice of unitaries that this measurement yields non-negligible correlations is extremely small. Then, we would like to discretize the space of all measurement superoperators and use the union bound to show that the probability that \emph{any} measurement superoperator yields non-negligible correlations is still very small. For this to work, the ``number'' of measurements has to be much smaller than the reciprocal of the probability of getting a bad $U$. However, the set of measurement superoperators cannot be discretized directly, since (among other things) they contain a potentially unbounded number of outputs. Hence, we will instead use the above argument on $(s,\eta)$-quasi-measurements, which \emph{can} be discretized easily (see Lemma \ref{sizeofnet}), and then show that the best measurement cannot beat the best $(s,\eta)$-quasi-measurement by too much. Along the way, we also prove the special case where the measurement device is constrained to making projective measurements, which can be viewed simply as $(CE',1)$-quasi-measurements.

The basic ingredient of the proof is the following concentration of measure theorem on Haar-distributed unitaries: 
\begin{thm}[Corollary 4.4.28 in \cite{zeitouni-random-matrices}]\label{thm:levy}
	Let $f : \mathcal{U}(d) \rightarrow \mbR$ be a function with Lipschitz constant $\theta$ (see Definition \ref{def:lipschitz}; the Lipschitz constant is taken with respect to the Hilbert-Schmidt distance on unitaries). Then,
	\[ \Pr_U \left\{ |f(U) - \mbE_{U} f | > \varepsilon  \right\} \leqslant \exp\left( -\frac{d \varepsilon^2}{4\theta^2} \right). \]
\end{thm}
We apply the theorem to the function 
\[ g_{\mathcal{M}}(U) = \left\| \mathcal{M}\left(\rho^{MCE'}\right) - \mathcal{M}\left(\rho^M \otimes \rho^{CE'}\right) \right\|_1 \]
for any fixed $(s,\eta)$-quasi-measurement $\mathcal{M}$, where $\rho$ depends on $U$ as in Equation (\ref{eqn:def-rho}). To do this, we need to bound two quantities from above: the expected value $\mbE g_{\mathcal{M}}(U)$ and the Lipschitz constant $\theta$. The bounds appear in Lemmas \ref{expfUunent} and \ref{lipfUunent} respectively, and the resulting concentration statement looks like (see Equation (\ref{eqn:bound-fixed-measurement})):
\[ \Pr_U \left\{ g_{\mathcal{M}}(U) > \varepsilon  \right\} \leqslant \exp\left( - \frac{2^{H_{\min}(M) + H_{\min}(E)}CKE}{2^6 \eta^2} \left( \varepsilon - \frac{2\Delta_{E,\infty}}{\sqrt{KE}} \right)^2 \right). \]
Now we are in a position to use our $\varepsilon$-net over $(s,\eta)$-quasi-measurements (Lemma \ref{sizeofnet}) and the union bound to get a bound on the probability that there exists an $(s,\eta)$-quasi-measurement $\mathcal{M}$ for which $g_{\mathcal{M}}(U) > \varepsilon$; this is done in Theorem \ref{thm:concentration}.

At this point, the proof splits into an ``easy'' and a ``hard'' branch. The easy branch (Section \ref{sec:projective-measurements}) applies Theorem \ref{thm:concentration} to projective measurements. The result is immediate, since a projective measurement is simply a $(CE',1)$-quasi-measurement. The hard branch (Section \ref{sec:povms}) goes for the full prize: showing that $g_{\mathcal{M}}(U)$ is small for every POVM with high probability. For this, we essentially show that a POVM corresponds (for the purposes of this proof) to a distribution over sequences of $s$ states. The operator Chernoff bound can then be used to show that this distribution is almost entirely supported on sequences that are $(s,\eta)$-quasi-measurements, for $s = O(CE \log(CE))$ and $\eta = O(1)$. We then apply Theorem \ref{thm:concentration} on these sequences, conditioned on the sequence being an $(s,\eta)$-quasi-measurement. A trivial bound is sufficient to cover the other case.

All that is then left to do to get the statements in the theorems stated above is to calculate conditions on the various parameters to make the exponent a reasonably large negative number.

\section{Concentration of the distinguishability from independence}\label{sec:preliminary-calculations}
To be able to use the general concentration of measure theorem (Theorem \ref{thm:levy}) on $g_{\mathcal{M}}(U)$, we must first be able to upper-bound the expectation of $g_{\mathcal{M}}(U)$ with respect to $U$. The following lemma does this:
\begin{lem}[Distinguishability for a fixed measurement]
	\label{expfUunent}
	If $\mcal{M}^{CE' \rightarrow X}$is an $(s,\eta)$-quasi-measurement, then
	\begin{equation*}
		\expU \lno \mcal{M} \lp \rho^{MCE'} \rp - \mcal{M} \lp \rho^{M} \ox \rho^{CE'} \rp  \rno_1 \leqslant \frac{2\Delta_{E, \infty}}{\sqrt{KE}}.
	\end{equation*}
\end{lem}
\begin{proof}
We begin by expanding and simplifying the original expression
\begin{eqnarray}
	\lefteqn{\expU \lno \mcal{M} \lp \rho^{MCE'} \rp - \mcal{M} \lp \rho^{M} \ox \rho^{CE'} \rp \rno_1 = \expU \lno \mcal{M} \lp \rho^{MCE'} - \rho^{M} \ox \rho^{CE'} \rp \rno_1} \\
	& \leqslant & \expU \sqrt{s \Tr \ls \lp (\sigma^M)^{-1/4} \: \mcal{M} \lp \rho^{MCE'} - \rho^{M} \ox \rho^{CE'} \rp \: (\sigma^M)^{-1/4} \rp^2 \rs } \nonumber \\
	& \leqslant & \sqrt{s \expU \Tr \ls \lp (\sigma^M)^{-1/4} \: \mcal{M} \lp \rho^{MCE'} - \rho^{M} \ox \rho^{CE'} \rp \: (\sigma^M)^{-1/4} \rp^2 \rs }. \label{eqn:preswapdef}
\end{eqnarray}
In the manipulations above we have used the linearity of the superoperator $\mcal{M}$ in the first line. In the second line we have used Lemma \ref{lem:tight12} with $\gamma = \ident^X \otimes \sigma^M$, noting that $|X|=s$. The third line follows from the concavity of the square root. We will now use a helpful identity for the trace of an operator squared: $\Tr Z^2 = \Tr ( Z \otimes Z)F$, where $F$ is defined as follows.
\begin{defin} \label{def:swap}
The \emph{swap operator} on $A^{\otimes 2}$, which is written as $A \otimes \overline{A}$, is the unique linear operator $F^A$ satisfying 
\begin{equation*}
	F^A \lp \ket{\psi}^A \ket{\phi}^{\overline{A}} \rp =  \ket{\phi}^A \ket{\psi}^{\overline{A}} \hspace{0.5in} \forall \ket{\psi}, \ket{\phi}.
\end{equation*}
\end{defin}

Expressing Equation (\ref{eqn:preswapdef}) using the swap operator gives
\begin{eqnarray}
	\lefteqn{\expU \lno \mcal{M} \lp \rho^{MCE'} \rp - \mcal{M} \lp \rho^{M} \ox \rho^{CE'} \rp \rno_1} \\
	& \leqslant & \sqrt{s \Tr \ls \expU \lp (\sigma^M)^{-1/4} \: \mcal{M} \lp \rho^{MCE'} - \rho^{M} \ox \rho^{CE'} \rp \: (\sigma^M)^{-1/4} \rp^{\ox 2} F^{XM} \rs } \nonumber \\
	&=& \sqrt{ \frac{(CE')^2}{s} \sum_{i=1}^s \Tr \ls \lp F^M \otimes \lp \chi_i^{CE'} \rp^{\ox 2} \rp  \expU \ls  (\sigma^M)^{-1/4}  ( \rho^{MCE'} - \rho^{M} \ox \rho^{CE'} ) (\sigma^M)^{-1/4} \rs^{\otimes 2} \rs } 
	 \label{eqn:tracedistanceF} 
\end{eqnarray}
Equation (\ref{eqn:tracedistanceF}) follows from the fact that results of the measurement $\mcal{M}$ are stored in an orthonormal basis of system $X$. We will proceed by evaluating $ \expU ( (\sigma^M)^{-1/4} \: ( \rho^{MCE'} - \rho^{M} \ox \rho^{CE'} ) \:  (\sigma^M)^{-1/4} )^{\ox 2}$, but before continuing we absorb the two $\sigma^{-1/4}$ into the operator $\rho$. That is we define,
\begin{eqnarray} 
	\tilde{\sigma}^{MCK} & := & \sum_{m=1}^{|M|} \sqrt{p_m} \ketbra{m}^M \otimes \ketbra{\psi_m}^{CK}  \label{def:sigmagamma} \quad \mbox{and} \\
	\tilde{\rho}^{MCKEE'} & := &  (\sigma^M)^{-1/4} \: \rho^{MCKEE'} \: (\sigma^M)^{-1/4} =\lp \ident^{ME} \ox U^{CKE} \rp \cdot \lp \tilde{\sigma}^{MCK} \ox \omega^{EE'} \rp. \label{def:rhogamma} 
\end{eqnarray}
With these two definitions in hand we can expand $\expU \lp \tilde{\rho}^{MCE'} - \tilde{\rho}^{M} \ox \rho^{CE'} )^{\otimes 2}\rp$ as
\begin{eqnarray}
	\lefteqn{\expU \lp \tilde{\rho}^{MCE'} - \tilde{\rho}^{M} \ox \rho^{CE'} \rp^{\ox 2}} \\
	& = & \expU \lp \Tr_{KE} \ls U^{CKE} \cdot \lp \lp \tilde{\sigma}^{MCK} - \tilde{\sigma}^{M} \ox \sigma^{CK}\rp \ox \omega^{EE'} \rp \rs \rp^{\ox 2} \nonumber \\
	& = & \Tr_{KE\overline{KE}} \ls \expU  \lp U^{CKE} \cdot \lp \lp \tilde{\sigma}^{MCK} - \tilde{\sigma}^{M} \ox \sigma^{CK} \rp \ox \omega^{EE'} \rp \rp^{\ox 2} \rs \nonumber \\
	& = & \Tr_{KE\overline{KE}} \ls \int \lp U^{CKE} \ox U^{\overline{CKE}} \ox \ident^{ME'\overline{ME'}} \rp \cdot \lp \lp \tilde{\sigma}^{MCK} - \tilde{\sigma}^{M} \ox \sigma^{CK} \rp \ox \omega^{EE'} \rp^{\ox 2} \mathrm{d}U \rs. \label{eqn:integralsource}
\end{eqnarray}

To evaluate the integral with Lemma \ref{Schur}, we will need to calculate the projections of our operator onto the symmetric and antisymmetric subspaces of $(CKE)^{\ox 2}$. Since the projectors onto the symmetric and antisymmetric subspaces can be written as $\Pi_{\pm} = \demi (\ident \pm F)$, we can arrive at same results by working with $\ident$ and $F$. We begin with $\ident$:
\begin{eqnarray}
	\lefteqn{\Tr_{CKE\overline{CKE}} \ls \lp \tilde{\sigma}^{MCK} - \tilde{\sigma}^{M} \ox \sigma^{CK} \rp^{\ox 2} \ox \lp \omega^{EE'} \rp^{\ox 2} \ident^{CKE\overline{CKE}} \rs} \label{eqn:schurweylI} \\
	& = & \sum_m \sqrt{p_m} \ketbra{m}^M \ox \sum_{m'} \sqrt{p_{m'}} \ketbra{m'}^{\overline{M}} \Tr_{CK} \ls \psi_m - \sum_{m''} p_{m''} \psi_{m''} \rs^2 \ox \lp \omega^{E'} \rp^{\ox 2} \nonumber \\
	& = & (1-1)^2 \cdot (\tilde{\sigma}^M \ox \omega^{E'} )^{\ox 2} = 0 \nonumber.
\end{eqnarray}
The projection onto $F$ requires a more subtle calculation,
\begin{eqnarray}
	\lefteqn{\Tr_{CKE\overline{CKE}} \ls \lp \tilde{\sigma}^{MCK} - \tilde{\sigma}^{M} \ox \sigma^{CK} \rp^{\ox 2} \ox \lp \omega^{EE'} \rp^{\ox 2} F^{CKE} \rs} \label{eqn:schurweylF} \\
	& = & \sum_m \sqrt{p_m} \ketbra{m}^M \ox \sum_{m'} \sqrt{p_{m'}} \ketbra{m'}^{\overline{M}} \nonumber \\
	& & \;\;\; \cdot \Tr_{CK} \ls \lp \psi_m - \sum_{m''} p_{m''} \psi_{m''} \rp \lp \psi_{m'} - \sum_{m'''} p_{m'''} \psi_{m'''} \rp \rs \ox \Tr_{E\overline{E}} \ls \lp \omega^{EE'} \rp^{\ox 2} F^{E} \rs \nonumber.
\end{eqnarray}

By taking a closer look at Equation (\ref{eqn:schurweylF}) we can make the simplification
\begin{eqnarray}
	\lefteqn{\Tr_{CK} \ls \lp \psi_m - \sum_{m''} p_{m''} \psi_{m''} \rp \lp \psi_{m'} - \sum_{m'''} p_{m'''} \psi_{m'''} \rp \rs} \nonumber \\
	& = & \Tr_{CK} \ls \psi_m \psi_{m'} - \sum_{m''} p_{m''} \psi_{m''} \psi_{m'} - \sum_{m'''} p_{m'''} \psi_m \psi_{m'''} + \sum_{m'', m'''} p_{m''} p_{m'''} \psi_{m''} \psi_{m'''} \rs \nonumber \\
	& = & \delta_{mm'} - p_{m'} - p_m + \sum_{m''} p_{m''}^2. \label{eqn:psimreductionresult}
\end{eqnarray}

Now we define $\tilde{\sigma}_{\circ}^{M\overline{M}}$ as the quantity evaluated in Equation (\ref{eqn:schurweylF}). Substituting the result of Equation  (\ref{eqn:psimreductionresult}) gives
\begin{eqnarray*}
	\tilde{\sigma}_{\circ}^{M\overline{M}} & : = & \Tr_{CK\overline{CK}} \ls \lp \tilde{\sigma}^{MCK} - \tilde{\sigma}^{M} \ox \sigma^{CK} \rp^{\ox 2} F^{CK} \rs \\
	& = & \lp \sum_m p_m \lp \ketbra{m} \rp^{\ox 2} - \tilde{\sigma}^M \ox (\tilde{\sigma}^{\overline{M}})^3  - (\tilde{\sigma}^M)^3 \ox \tilde{\sigma}^{\overline{M}} +  \lp \sum_m p_m^2 \rp \tilde{\sigma}^{M} \ox \tilde{\sigma}^{\overline{M}} \rp.
\end{eqnarray*}

We also define $\Omega^{E'\overline{E'}}$ as the operator acting on system $E'\overline{E'}$ in Equation (\ref{eqn:schurweylF}), or $\Omega^{E'\overline{E'}} = \Tr_{E\overline{E}} [( \omega^{EE'} )^{\ox 2} F^{E} ]$. At this point, Lemma \ref{Schur} can be used to evaluate the integral in Equation (\ref{eqn:integralsource}). We can make significant simplifications by first expanding the $\alpha_{\pm}$ and then using our result from Equation (\ref{eqn:schurweylI}) to show that 
\begin{eqnarray*}
	\alpha_{\pm} & = & \frac{1}{\mathrm{rank}(\Pi^{CKE}_{\pm})} \Tr_{CKE\overline{CKE}} \ls \lp \tilde{\sigma}^{MCK} - \tilde{\sigma}^{M} \ox \sigma^{CK} \rp^{\ox 2} \ox \lp \omega^{EE'} \rp^{\ox 2} \lp \Pi^{CKE}_{\pm} \ox \ident^{ME'} \rp \rs \\
	& = & \frac{\pm \lp \tilde{\sigma}_{\circ}^{M\overline{M}} \ox \Omega^{E'\overline{E'}} \rp}{CKE(CKE\pm1)},
\end{eqnarray*}
where the terms $\Pi^{CKE}_{\pm}$ are the projectors onto the symmetric and antisymmetric subspaces of $(CKE)^{\ox 2}$, that is $\demi (\ident^{CKE\overline{CKE}} \pm F^{CKE})$. In particular, because $\alpha_+$ is proportional to $\alpha_-$, the integral will have the product form
\begin{equation*}
	\tilde{\sigma}_{\circ}^{M\overline{M}} \ox \Omega^{E'\overline{E'}} \ox \lp \frac{\Pi^{CKE}_+}{CKE(CKE+1)} -  \frac{\Pi^{CKE}_-}{CKE(CKE-1)} \rp,
\end{equation*}
so the calculation of the trace in Equation (\ref{eqn:tracedistanceF}) will factor into a product
over the systems $(M)^{\ox 2}$ and $(CKEE')^{\ox 2}$. Thus,
\begin{eqnarray}
	\lefteqn{\Tr_{(MKE)^{\ox 2}} \ls \lp \Tr_{(CE')^{\ox 2}} \ls  \lp \chi_i^{CE'} \rp^{\ox 2} \expU \lp \tilde{\rho}^{MCE'} - \tilde{\rho}^{M} \ox \rho^{CE'} \rp^{\ox 2} \rs \rp F^{M} \rs} \label{eqn:factorizedint} \\
	& = & \Tr \ls \tilde{\sigma}_{\circ}^{M\overline{M}}  F^{M} \rs \cdot \Tr \ls \lp \chi_i^{CE'} \ox \ident^{KE} \rp^{\ox 2} \lp \frac{\Pi^{CKE}_+ \ox \Omega^{E'\overline{E'}}}{CKE(CKE+1)} -  \frac{\Pi^{CKE}_- \ox \Omega^{E'\overline{E'}}}{CKE(CKE-1)} \rp \rs. \nonumber
\end{eqnarray}
The first first factor in Equation (\ref{eqn:factorizedint}) can easily be bounded:
\begin{eqnarray}
	\Tr_{M\overline{M}} \ls \tilde{\sigma}_{\circ}^{M\overline{M}}  F^{M} \rs & = & \sum_m p_m - \sum_m p_m^{3/2} - \sum_m p_m^{3/2} + \sum_m p_m^2  \nonumber \\
	& \leqslant & 2 \sum_m p_m \nonumber 
	=  2. \label{eqn:factorizedone}
\end{eqnarray}

To estimate the second factor in Equation (\ref{eqn:factorizedint}) we will need to observe two facts. First, that
\begin{equation}
	\Tr \ls ( \chi_i^{CE'} \ox \ident^{KE} )^{\ox 2} \: \ident^{CKE\overline{CKE}} \ox \Omega^{E'\overline{E'}} \rs \leqslant (KE)^2 \lno \Omega^{E'\overline{E'}} \rno_{\infty}, \label{eqn:finalexppartone}
\end{equation}
which follows from the fact that $\chi_i^{CE'}$ is a rank $1$ projector. Second, that
\begin{eqnarray}
	\Tr \ls ( \chi_i^{CE'} \ox \ident^{KE} )^{\ox 2} \: F^{CKE} \ox \Omega^{E'\overline{E'}} \rs & = & KE \; \Tr_{E'\overline{E'}} \ls \lp \Tr_{C\overline{C}} \ls (\chi_i^{CE'})^{\ox2} F^C \rs \rp \Omega^{E'\overline{E'}} \rs \nonumber \\
	& \leqslant & KE \lno \Omega^{E'\overline{E'}} \rno_{\infty}. \label{eqn:finalexpparttwo}
\end{eqnarray}
If we use Equations (\ref{eqn:finalexppartone}) and (\ref{eqn:finalexpparttwo}) to estimate the second factor of Equation (\ref{eqn:factorizedint}) we get the bound
\begin{eqnarray}
	\lefteqn{\Tr \ls \lp \chi_i^{CE'} \ox \ident^{KE} \rp^{\ox 2} \lp \frac{\Pi^{CKE}_+ \ox \Omega^{E'\overline{E'}}}{CKE(CKE+1)} -  \frac{\Pi^{CKE}_- \ox \Omega^{E'\overline{E'}}}{CKE(CKE-1)} \rp \rs} \nonumber \\
	& \leqslant & \lp \frac{(KE)^2 + KE}{2CKE(CKE+1)} - \frac{(KE)^2 - KE}{2CKE(CKE-1)} \rp \cdot \lno \Omega^{E'\overline{E'}} \rno_{\infty} \nonumber \\
	& \leqslant & \frac{2}{C^2KE} \cdot \lno \Omega^{E'\overline{E'}} \rno_{\infty} \label{eqn:factorizedtwo} .
\end{eqnarray}
This can be rewritten in a more familiar form using
\begin{eqnarray*}
	 \lno \Omega^{E'\overline{E'}} \rno_{\infty} & = & \lno \Tr_{E'\overline{E'}} \ls \lp \omega^{EE'} \rp^{\ox 2} F^{E} \rs \rno_{\infty} \\
	& = & \lno \lp \omega^{E} \rp^{\ox 2} F^{E} \rno_{\infty} = \lno \omega^E \rno^2_{\infty} = 2^{-2H_{\mathrm{min}}(E)_{\omega}}.
\end{eqnarray*}
In the above, the third equality follows from the fact that the operator norm is right-invariant under unitary transformations and $F$ is a unitary matrix. Combining the results in Equations (\ref{eqn:factorizedone}) and (\ref{eqn:factorizedtwo}), as well as the above identity, we obtain an upper bound for the trace distance through Equation (\ref{eqn:tracedistanceF}),
\begin{eqnarray*}
	\expU \lno \mcal{M} \lp \rho^{MCE'} \rp - \mcal{M} \lp \rho^{M} \ox \rho^{CE'} \rp \rno_1 & \leqslant & \sqrt{s \lp \frac{CE'}{s}\rp^2  \sum_{i=1}^s 2 \frac{2 \cdot 2^{-2H_{\mathrm{min}}(E)_{\omega}}}{(C)^2KE}} \\
	& \leqslant & \frac{2\Delta_{E, \infty}}{\sqrt{KE}}.
\end{eqnarray*}
\end{proof}

\begin{lem}
	\label{lipfUunent}
	$g_{\mcal{M}}(U)$, the trace distance to independence for a fixed $(s,\eta)$-quasi-measurement, is Lipschitz continuous on the space $(\mcal{U}(CKE), \lno \cdot \rno_2)$ with constant $4\eta \sqrt{\Delta_{M, \infty} \; \Delta_{E, \infty} / ME}$.
\end{lem}
\begin{proof}
We wish to analyze the behaviour of the trace distance with respect to the unitary matrix defining the channel. Recall the definition of function $g_{\mcal{M}}(U)$,
\begin{equation*}
	g_{\mcal{M}}(U) =  \lno \mcal{M} \lp \rho^{MCE'} \rp - \mcal{M} \lp \rho^{M} \ox \rho^{CE'} \rp \rno_1. 
\end{equation*}
If we denote by $\rho_U$ and $\rho_V$ the states $\Tr_{K} \ls U \cdot \sigma \rs$ and $\Tr_{K} \ls V \cdot \sigma \rs$ respectively, we can bound the deviation of $g_{\mcal{M}}$ using the triangle inequality by 
\begin{eqnarray}
	\left| g_{\mcal{M}}(U) - g_{\mcal{M}}(V) \right| & \leqslant & \lno \mcal{M} \lp \rho_U^{MCE'} \rp - \mcal{M} \lp \rho_V^{MCE'} \rp \rno_1 \nonumber \\
	& & \hspace{0.5cm} +  \lno \mcal{M} \lp \rho_U^{M} \ox \rho_U^{CE'} \rp - \mcal{M} \lp \rho_V^{M} \ox \rho_V^{CE'} \rp \rno_1 \label{eqn:Lipun:1}\\
	& = & \lno \mcal{M} \lp \rho_U^{MCE'} - \rho_V^{MCE'} \rp \rno_1 +  \lno \mcal{M} \lp \sigma^{M} \ox \lp \rho_U^{CE'} - \rho_V^{CE'} \rp \rp \rno_1 \nonumber,
\end{eqnarray}
where the second line follows from the linearity of the superoperator. We note that for any hermitian operator $\zeta$,
\begin{eqnarray*}
	\lno \mcal{M} \lp \zeta \rp \rno_1 & = & \lno \frac{CE'}{s} \sum_{i=1}^s \ketb{i}{\chi_i} \zeta \ketb{\chi_i}{i} \rno_1 \\
	& = & \frac{CE'}{s} \sum_{i=1}^s \left| \bra{\chi_i} \zeta \ket{\chi_i} \right| \leqslant \frac{CE'}{s} \sum_{i=1}^s \bra{\chi_i} | \zeta | \ket{\chi_i} \\
	& = & \frac{CE'}{s} \sum_{i=1}^s \Tr \ls \chi_i |\zeta| \rs \leqslant \eta \lno \zeta \rno_1,
\end{eqnarray*}
where the last inequality follows from the definition of $(s,\eta)$-quasi-measurements. Applying this new fact, our bound in Equation (\ref{eqn:Lipun:1}) becomes,
\begin{eqnarray}
	\left| g_{\mcal{M}}(U) - g_{\mcal{M}}(V) \right| & \leqslant & \eta \lno \rho_U^{MCE'} - \rho_V^{MCE'} \rno_1 +  \eta \lno \sigma^{M} \ox \lp \rho_U^{CE'} - \rho_V^{CE'} \rp \rno_1 \label{eqn:Lipun:2} \\
	& \leqslant & 2\eta \lno \rho_U^{MCKEE'} - \rho_V^{MCKEE'} \rno_1 \nonumber \\ 
	& = & 2\eta \lno \lp U-V \rp \cdot \sigma \ox \omega \rno_1, \nonumber
\end{eqnarray}
where the second line follows from monotonicity. We introduce a purification of $\sigma^{MCK}$ in a new but temporary system $N$ such that $\mathrm{dim}(N)=\mathrm{dim}(M)$. We also recall that $\omega$ is pure. This permits us to use Lemma \ref{1normto2norm} and arrive at the following consequence of Equation (\ref{eqn:Lipun:2}),
\begin{equation}
	\left| g_{\mcal{M}}(U) - g_{\mcal{M}}(V) \right| \leqslant 4\eta \lno \lp U^{CKE} - V^{CKE}\rp \ox \ident_{MNE'} \ket{\sigma}^{MNCK}\ket{\omega}^{EE'} \rno_2 . \label{eqn:Lipun:3}
\end{equation}
We now introduce a helpful operation.
\begin{defin}[Vector-operator correspondence] \label{def:opAB}
	Endow systems $A$ and $B$ with fixed orthonormal bases $\{ \ket{a_i}^A \}_i$ and $\{ \ket{b_i}^B \}_i$ respectively, and let $\mathrm{op}_{A \rightarrow B} : A \ox B \rightarrow \mathrm{L}(A,B)$, the space of linear transformations from $A$ to $B$, be defined as
\begin{equation*}
	\mathrm{op}_{A \rightarrow B} \lp \ket{a_i} \ket{b_j} \rp = \ket{b_j} \bra{a_i} \hspace{0.5in} \forall i,j
\end{equation*}
This operation depends on the choice of basis; therefore, whenever it is used, a particular choice of basis is implied. Since this choice will never matter in our calculations, we shall not explicitly define these bases.
\end{defin}
Useful properties of the correspondence can be found in~\cite{fred-these}.

We can think of the operator $ (U^{CKE} - V^{CKE}) \ox \ident^{MNE'}$ as bipartite over composite systems $MNE'$ and $CKE$. Since the $2$-norm depends only on the Schmidt coefficients of the states, it will be invariant under the $\mathrm{op}$ operation defined in Definition \ref{def:opAB}. Our bound from Equation (\ref{eqn:Lipun:3}) then becomes,
\begin{eqnarray*}
	\left| g_{\mcal{M}}(U) - g_{\mcal{M}}(V) \right| & \leqslant & 4\eta \lno \mathrm{op}_{MNE' \rightarrow CKE} \lp \lp U^{CKE} - V^{CKE}\rp \ox \ident^{MNE'} \ket{\sigma}^{MNCK} \ket{\omega}^{EE'} \rp \rno_2 \\
	& = & 4\eta \lno \lp U - V \rp \mathrm{op}_{MNE' \rightarrow CKE} \lp \ket{\sigma}\ket{\omega} \rp \rno_2, 
\end{eqnarray*}
where the second line follows from the fact that $\mathrm{op}_{MNE' \rightarrow CKE}$ is linear and commutes with unitary transformations on $CKE$. We are left with a few easy steps to bound the Lipschitz constant. 
\begin{eqnarray*}
	\left| g_{\mcal{M}}(U) - g_{\mcal{M}}(V) \right| & \leqslant & 4\eta \lno U - V \rno_2 \lno \mathrm{op}_{MNE' \rightarrow CKE} \lp \ket{\sigma} \ket{\omega} \rp \rno_{\infty} \\
	& = & 4\eta \lno U - V \rno_2 \sqrt{\lno \sigma^{CK} \ox  \omega^{E} \rno_{\infty}} \\
	& = & 4\eta \lno U - V \rno_2 \sqrt{\lno \sum_m p_m \ketbra{\psi_m}^{CK}\rno_{\infty} \lno \omega^E \rno_{\infty}} \\
	& = & 4\eta \lno U - V \rno_2 \sqrt{\mathrm{max} \; p_m \; \cdot 2^{-H_{\mathrm{min}}(E)_{\omega}}} \\
	& = & 4\eta \lno U - V \rno_2 2^{-\tfrac{1}{2} H_{min}(M)_{\sigma}} 2^{- \tfrac{1}{2} H_{\mathrm{min}}(E)_{\omega}}\\
	& = & \frac{4\eta \sqrt{\Delta_{M, \infty} \; \Delta_{E, \infty} }}{\sqrt{ME}} \lno U - V \rno_2.
\end{eqnarray*}
A proof of the inequality can be found, for example, in~\cite{fred-these}.
The second line follows from the fact the Schmidt coefficients of $\ket{\sigma}^{MNCK}$ are the square roots of the eigenvalues of $\sigma^{CK}$. The last line follows from the definition of $\Delta_{\mathrm{min}}$.
\end{proof}

In order to discretize the set of all $(s,\eta)$-quasi-measurements, we require a distance measure for the set.
\begin{defin}[Metric on the set of $(s,\eta)$-quasi-measurements, $\mcal{L}(s,\eta)$]
	\label{definDistonL}
Consider $\mcal{M}$, $\mcal{N}$ $\in \mcal{L}(s,\eta)$ defined as
\begin{equation*}
	\mcal{M} \lp \sigma \rp = \frac{|CE'|}{s} \sum_{i=1}^{s} \ketb{i}{ \chi_i } \sigma \ketb{ \chi_i }{i} , \hspace{0.5in} \mcal{N} \lp \sigma \rp = \frac{|CE'|}{s} \sum_{i=1}^{s} \ketb{i}{ \nu_i } \sigma \ketb{ \nu_i }{i}.
\end{equation*}
We define the distance between these two elements as
\begin{equation*}
	d(\mcal{M},\mcal{N}) := \sum_{i=1}^s \lno \chi_i - \nu_i \rno_2.
\end{equation*}
\end{defin}

Now letting $\mcal{M}$ vary instead of $U$, we define a new function $h_{U}(\mcal{M})$ by
\begin{equation*}
	h_{U}(\mcal{M}) = \lno \mcal{M} \lp \rho^{MCE'} \rp - \mcal{M} \lp \rho^{M} \ox \rho^{CE'} \rp \rno_1.
\end{equation*}

\begin{lem}
	\label{lipfLunent}
	$h_{U}(\mcal{M})$  is Lipschitz continuous on the space $(\mcal{L}(s,\eta), d)$ with constant $\frac{2\sqrt{CE'}}{s} \sqrt{\Delta_{M,2} \Delta_{E,2}}$.
\end{lem}
\begin{proof}
As for Lemma \ref{lipfUunent}, we can use the triangle inequality to rewrite the variation of the trace distance as follows,
\begin{eqnarray}
	\lefteqn{|h_{U}(\mcal{M}) - h_{U}(\mcal{N})|} \nonumber \\
	& \leqslant & \lno \mcal{M} \lp \rho^{MCE'} \rp - \mcal{N} \lp \rho^{MCE'} \rp \rno_1 +  \lno \mcal{M} \lp \rho^{M} \ox \rho^{CE'} \rp - \mcal{N} \lp \rho^{M} \ox \rho^{CE'} \rp \rno_1 \nonumber \\
	& = & \frac{CE'}{s} \sum_{i=1}^n \lp \lno \Tr_{CE'} \ls \lp \chi_i^{CE'} - \nu_i^{CE'} \rp \rho^{MCE'} \rs \rno_1 +  \lno \Tr_{CE'} \ls \lp \chi_i^{CE'} - \nu_i^{CE'} \rp \rho^{M} \ox \rho^{CE'} \rs \rno_1 \rp \nonumber \\
	& \leqslant & \frac{CE'}{s} \sum_{i=1}^s \lno \lp \chi_i^{CE'} - \nu_i^{CE'} \rp \rho^{MCE'} \rno_1 + \frac{CE'}{s} \sum_{i=1}^s \lno \lp \chi_i^{CE'} - \nu_i^{CE'} \rp \rho^{M} \ox \rho^{CE'} \rno_1 \nonumber \\
	& \leqslant & \frac{CE'}{s} \sum_{i=1}^s \lno \chi_i^{CE'} - \nu_i^{CE'} \rno_2 \lno \rho^{MCE'} \rno_2 + \frac{CE'}{s} \sum_{i=1}^s \lno \chi_i^{CE'} - \nu_i^{CE'} \rno_2 \lno \rho^{M} \ox \rho^{CE'} \rno_2, \label{eqn:LipM:1}
\end{eqnarray}
where the last line follows from the operator version of the Cauchy-Schwarz inequality (see Equation (IX.32) in \cite{bhatia}). Consider momentarily the second factor in the first term in Equation (\ref{eqn:LipM:1}), 
\begin{eqnarray}
	\lno \rho^{MCE'} \rno_2 & = & \lno \Tr_{KE} \ls U_{CKE} \cdot \lp \sigma^{MCK} \ox \omega^{EE'} \rp \rs \rno_2 \nonumber \\
	& = & \sqrt{\Tr \ls \lp U_{CKE}^{\ox 2} \cdot \lp \sigma^{MCK} \ox \omega^{EE'} \rp^{\ox 2} \rp F^{MCE'} \rs } \nonumber \\
	& = & \sqrt{\Tr \ls \lp U_{CKE}^{\ox 2} \; F^C \; U^{\dag \ox 2}_{CKE}\rp \lp \sum_m p_m^2 (\ketbra{\psi_m}^{CK})^{\ox 2} \ox (\omega^{EE'})^{\ox 2} \rp  F^{E'} \rs } \nonumber \\
	& \leqslant & \sqrt{\Tr \ls \sum_m p_m^2 (\ketbra{\psi_m}^{CK})^{\ox 2} \ox (\omega^{EE'})^{\ox 2} F^{E'} \rs } \nonumber \\
	& = & \sqrt{\Tr \ls (\omega^{EE'})^{\ox 2} F^{E'} \rs \sum_m p_m^2 } = 2^{-\demi H_2(M)_{\sigma} - \demi H_2(E)_{\omega}}. \label{eqn:LipM:2}
\end{eqnarray}
The third line is true by the cyclic property of the trace. The inequality, however, is true by the following observation: since $F^2=\ident$ we know that $F$ has eigenvalues $\pm 1$ and so $F\leq \ident$. We can make a similar evaluation for the last factor in Equation (\ref{eqn:LipM:1}), 
\begin{equation}
	\lno \rho^{M} \ox \rho^{CE'} \rno_2 \leqslant 2^{-\demi H_2(M)_{\sigma} - \demi H_2(E)_{\omega}} \label{eqn:LipM:3} ,
\end{equation}
since this inequality is a just a special case of the calculations leading to Equation (\ref{eqn:LipM:2}). If we apply Equations (\ref{eqn:LipM:2}) and (\ref{eqn:LipM:3}) to Equation (\ref{eqn:LipM:1}), we can extract a very simple bound on the variation of the trace distance
\begin{eqnarray*}
	|h_{U}(\mcal{M}) - h_{U}(\mcal{N})| & \leqslant & \frac{2CE'}{s} 2^{-\demi H_2(M)_{\sigma} - \demi H_2(E)_{\omega}}  \sum_{i=1}^s \lno \chi_i^{C} - \nu_i^{C} \rno_2 \\
	& \leqslant & \frac{2\sqrt{CE'}}{s} \sqrt{\Delta_{M,2} \Delta_{E,2}}  \: d(\mcal{M}, \mcal{N}), 
\end{eqnarray*}
where the last line follows from the definition of our metric on $\mcal{L}(s,\eta)$. We have also ignored a factor of $1/\sqrt{K}$ above when expressing the bound in terms of $\Delta_{M,2}$. We do this to simplify future calculations and it only gives a slightly less tight bound here. 
\end{proof}
\begin{lem}
\label{sizeofnet}
Given system $A$, there exists a $\varepsilon$-net $\mcal{J}$ over the set  $\mcal{L}(s,\eta)$ of all $(s,\eta)$-quasi-measurements on $A$, such that each element $L \in \mcal{L}(s,\eta)$ is at most $\varepsilon$-distant from an element of $J \in \mcal{J}$ with respect to the metric $d(\cdot, \cdot)$. The size of this net can be taken to be
\begin{equation*}
	|\mcal{J}| \leqslant \lp \frac{10 s}{\varepsilon} \rp^{2s|A|}.
\end{equation*}
\end{lem}

\begin{proof}
We begin by consider an $\varepsilon$-net $\mcal{K}$ over $\mbS_{2|A|}^{\times s}$ ($s$-tuples of $2|A|$-dimensional Euclidean unit spheres). First, there exists a $\varepsilon$-net over $\mbS_{2|A|}$ of size no more than $(5/\varepsilon)^{2|A|}$. (See, for example, Lemma II.4 in \cite{hlsw03}.) $\mcal{K}$ can then be constructed by assembling the direct product of all the nets on the individual unit spheres. This produces a new net on the set of $s$-tuples of $2|A|$-dimensional unit spheres. Recall the distance measure $d(\cdot, \cdot)$ over $\mcal{L}(s,\eta)$, the set of all $(s,\eta)$-quasi-measurements. This metric can be extended to $s$-tuples. If it is then evaluated for any $s$-tuple $x$ and its representative in the net $y$,
\begin{equation*}
	d(x, y) = \sum_{i=1}^s \lno \chi_i -\nu_i \rno_2  \leq s \varepsilon.
\end{equation*}
Thus the spacing of the net $\mcal{K}$ over $s$-tuples is at most $s\varepsilon$ with respect to the desired metric. Consider the following set:
\begin{equation*}
	\mcal{K}' := \lb y \in \mcal{K} : \exists \: x \in \mcal{L}(s,\eta) , \lno x-y \rno_2 \leqslant s\varepsilon \rb.
\end{equation*}
This is the set of all elements of the net $\mcal{K}$ which are close to $(s,\eta)$-quasi-measurements. In other words, all $(s,\eta)$-quasi-measurements use an element of $\mcal{K}'$ as their ``representative" in the net. Now, divide $\mcal{L}(s,\eta)$ into subsets of elements which share the same representation in $\mcal{K}'$ and construct $\mcal{J}$ by choosing one $L \in \mcal{L}$ from each subset. We then have by the triangle inequality that all $L \in \mcal{L}$ are $2 s \varepsilon$ close to their new representative in $\mcal{J}$. Clearly $|\mcal{J}| \leqslant |\mcal{K}|$ since it was constructed from a subset and if we wish to make an $\varepsilon$-net over $\mcal{L}(s,\eta)$ we need only rescale the $\varepsilon$ from above, giving the result.
\end{proof}
The Lipschitz constants, expectation value and net size give us all the pieces we need to make the concentration argument. We show that with very high probability, the distinguishability from independence of the joint (potentially unnormalized) distribution of messages and quasi-measurement outcomes  is small. 
\begin{thm}[Concentration of probability for distinguishability from independence] \label{thm:concentration}
Given the quantum state $\rho^{MCKEE'} = U^{CKE} \cdot (\sigma^{MCK} \ox \omega^{EE'})$ where $U$ is a random unitary operator chosen according to the Haar measure, $\sigma$ is as defined in Equation (\ref{eqn:sigmadef}), $E' \cong E$, and $\omega^{EE'}$ is a bipartite pure state, the following bound holds
\begin{multline*}
	\PrU \lb \underset{\mcal{M} \in \mathcal{L}(s,\eta)}{\mathrm{sup}} \lno \mcal{M} \lp \rho^{MCE'}\rp - \mcal{M} \lp \rho^M \ox \rho^{CE'} \rp \rno_1 >  \varepsilon \rb \\
\leqslant \mathrm{exp} \lp 2sCE \ln \lp\frac{40\sqrt{CE}}{ \varepsilon} \sqrt{\Delta_{M,2}\Delta_{E,2}} \rp - \frac{(CKE)^2}{2^8\eta^2 \Delta_{M, \infty} \Delta_{E, \infty}} \lp \varepsilon - \frac{4 \Delta_{E, \infty}}{\sqrt{KE}} \rp^2 \rp.
\end{multline*}
In the above, $\Delta_{M, \infty}$, $\Delta_{M,2}$, $\Delta_{E,2}$ and $\Delta_{E, \infty}$ are as defined in Equations (\ref{eqn:deltaMmindef}), (\ref{eqn:deltaM2def}), (\ref{eqn:deltaE2def}) and (\ref{eqn:deltaEmindef}). 
\end{thm}
\begin{proof}
We apply Theorem~\ref{thm:levy} to $g_{\mcal{M}}$ and consider only one direction of the divergence from the expected value. The exact statement can be written as
\begin{equation}\label{eqn:bound-fixed-measurement}
	\PrU \lb g_{\mcal{M}}(U) > \varepsilon \rb \leqslant \mathrm{exp} \lp - \frac{MCKE^2}{64\eta^2 \Delta_{M, \infty} \Delta_{E, \infty} } \lp \varepsilon - \expU g_{\mcal{M}} \rp^2\rp.
\end{equation}
It is convenient to define
\begin{equation*}
f(\mathcal{M},U) = \left\| \mathcal{M}\left(\rho^{MCE'}\right) - \mathcal{M}\left(\rho^M \otimes \rho^{CE'}\right) \right\|_1.
\end{equation*}
Clearly, $g_\mathcal{M}$ and $h_U$ are sections of $f$ and we are interested in bounding
$\PrU \{ \underset{\mcal{M}}{\mathrm{sup}} \; f(\mcal{M},U) >  \varepsilon \}$.
Let
\begin{equation*}
	\varepsilon' = \frac{s \varepsilon}{2\sqrt{CE\Delta_{M,2}\Delta_{E,2}}}, 
\end{equation*}
and consider $\mcal{J}$ an $\varepsilon'$-net over all $(s,\eta)$-quasi-measurements $\mcal{M}$. We found in Lemma \ref{lipfLunent} that if two $(s,\eta)$-quasi-measurements were $\varepsilon'$ apart with respect to the distance measure $d(\cdot, \cdot)$, then for a fixed unitary $U$, the values of $f$ for each measurement would not differ by more than $\varepsilon$. Thus we can state that the supremum deviation of $f$ is not more than twice the maximum deviation found on measurements in the net, 
\begin{equation*}
	\PrU \lb \underset{\mcal{M}}{\mathrm{sup}} f(\mcal{M}, U) >  2\varepsilon \rb \leqslant \PrU \lb \underset{\mcal{M} \in \mcal{J}}{\mathrm{max}} \; f(\mcal{M},U) >   \varepsilon \rb.
\end{equation*}
A union bound argument now bounds the probability of deviation for the maximum measurement by the probability of deviation for a generic measurement, 
\begin{equation*}
	\PrU \lb \underset{\mcal{M} \in \mcal{J}}{\max} f(\mcal{M},U) >  \varepsilon \rb \leqslant \sum_{\mcal{M} \in \mcal{J}} \PrU \lb g_{\mcal{M}}(U) >  \varepsilon \rb.
\end{equation*}
Thankfully, we have an explicit bound for the probability of deviation for an arbitrary measurement and we can make a simplification,
\begin{eqnarray*}
	\lefteqn{\PrU \lb \underset{\mcal{M}}{\mathrm{sup}} f(U, \mcal{M}) > 2 \varepsilon \rb \leqslant \sum_{\mcal{M} \in \mcal{J}} \mathrm{exp}  \lp - \frac{MCKE^2}{64\eta^2 \Delta_{M, \infty} \Delta_{E, \infty} } \lp \varepsilon - \expU f \rp^2\rp} \\
	& \leqslant & \lp \frac{20\sqrt{CE\Delta_{M,2}\Delta_{E,2}}}{ \varepsilon} \rp^{2sCE} \mathrm{exp} \lp - \frac{MCKE^2}{64\eta^2 \Delta_{M, \infty} \Delta_{E, \infty} } \lp \varepsilon - \expU f \rp^2\rp \\
	& \leqslant & \mathrm{exp} \lp 2sCE \ln \lp\frac{20\sqrt{CE\Delta_{M,2}\Delta_{E,2}}}{ \varepsilon} \rp - \frac{MCKE^2}{64\eta^2 \Delta_{M, \infty} \Delta_{E, \infty} } \lp \varepsilon - \expU f \rp^2\rp.
\end{eqnarray*}
Substituting in the fact that $CK=M$ yields the desired inequality.
\end{proof}

\section{Locking against projective measurements}\label{sec:projective-measurements}
In this section we will only consider projective measurements, in other words $(s,\eta) = (CE', 1)$. We will also state all of the subsequent theorems in terms of qubits. For this reason we will identify $C=2^c$, $K=2^k$ and $E=E'=2^e$. This last assumption, namely that $E$ and $E'$ have the same dimension, is crucial for this section because it restricts the size of the set of measurements sufficiently to allow for a straightforward discretization. The restriction will be lifted when we move on to generalized measurements in the next section, however. 

Our calculations, we will make repeated use of the fact that
\begin{equation}
 	\log(x+y) \leqslant x+ \log(y) \hspace{0.5in} \forall x,y \geqslant 1. \label{eqn:logrel}
\end{equation}
\begin{cor}[Locking for uniform messages with maximal entanglement] \label{cor:unihigh}
Consider the locking scheme described in Definition \ref{def:locking} for a uniform message with maximal entanglement available at the measurement. Choose $p$ and $\epsilon$ such that $\varepsilon > 8 \sqrt{1/KE}$ and $p > 2^{-2(CE)^2}$. Then the scheme will be an $\varepsilon$-locking locking scheme except with probability $p$ so long as the measurement superoperators are restricted to projective measurements and
\begin{equation*}
	k >  9+ 2 \log \frac{1}{\varepsilon} +\demi \log (c+e) \label{eqn:corunihigh}.
\end{equation*}
\end{cor}
\begin{proof}
Using Theorem \ref{thm:concentration}, we ensure that, except with probability $p$, our state is an $\varepsilon$-locking scheme provided that
\begin{equation*}
	2(CE)^2 \ln \lp \frac{40\sqrt{CE}}{\varepsilon} \rp - \frac{(CE)^2}{2^8} K^2 (\varepsilon')^2 < \ln p, 
\end{equation*}
where we've defined for the time being $\varepsilon'$ as $\varepsilon - 4/\sqrt{KE}$. A quick rearrangement of the terms reveals that the inequality will be satisfied if
\begin{equation}
	\frac{2^9}{(\varepsilon')^2} \ln \lp \frac{40\sqrt{CE}}{\varepsilon} \lp \frac{1}{p} \rp^{1/2(CE)^2} \rp < K^2 \label{eqn:unihighproof}.
\end{equation}
From our choice of $p$ we can easily see $(1/p)^{1/2(CE)^2}<2$ and from our choice of $\varepsilon$ we see that $2^9/(\varepsilon')^2 < 2^{13}/\varepsilon^2$. Thus inequality (\ref{eqn:unihighproof}) is satisfied when
\begin{equation*}
	\log \lp \frac{2^{13}}{\varepsilon^2} \rp + \log \lp \ln2 \log \lp \frac{80\sqrt{CE}}{\varepsilon}\rp \rp < 2k.
\end{equation*}
Finally, two applications of Equation (\ref{eqn:logrel}) reveal that the above is satisfied provided,
\begin{equation*}
	17 + 2 \log \frac{1}{\varepsilon} + \log \log \frac{1}{\varepsilon} + \log(c+e) < 2k.
\end{equation*}
Rearranging the terms we see that the above condition is satisfied provided inequality (\ref{eqn:corunihigh}) is satisfied, and we have completed the proof.
\end{proof}
Corollary \ref{cor:unihigh}, and its extension to arbitrary POVM measurements in Corollary \ref{cor:unihighPOVM} is a mathematical expression that ``generically, information is locked until it can be completely decoded.'' To arrive at this interpretation, recall from Equation (\ref{eqn:purified-decoding}) that to achieve a decoding error of $\epsilon$, the measurement must be supplied with the entanglement  through system $E'$ as well as a system $C$ satisfying $c - n > 2 \log(1/\epsilon)$. Of course, this condition could never be met if the constraint $n=c+k$ is assumed, but the constraint was only made for convenience to prove the locking results. Using it to re-express Corollary \ref{cor:unihigh}, though, we find that the information about the message is $\epsilon$-locked provided $c=n-k < n - 9 - 2 \log(1/\epsilon) - 1/2 \cdot \log(c+e)$. Therefore, regardless of the size of the message or the amount of entanglement, the message goes from being $\epsilon$-locked to being decodable with average probability of error at most $\epsilon$ with the transfer of $9 + 4 \log(1/\epsilon) + 1/2 \cdot \log(c+e)$ qubits.

At this point, we wish to study the dependence of the minimum key size $k$ on the various entropies of the message $M$ and the entanglement $E$.
\begin{cor}[Locking for messages of bounded entropy with imperfect entanglement] \label{cor:modmod}
Consider the locking scheme described in Definition \ref{def:locking} for a message of bounded entropy with entanglement of a bounded fidelity available at the measurement. Choose $\varepsilon$ and $p$ satisfying
\begin{equation*}
	\varepsilon >\frac{ 8 \Delta_{E,\infty}}{\sqrt{KE}}, \hspace{1.0in} p > 2^{-2(CE)^2}.
\end{equation*}
Then the scheme will be an $\varepsilon$-locking locking scheme except with probability $p$ so long as the measurement superoperators are restricted to projective measurements and
\begin{equation}
	k' + \demi \Big( n - H_{\mathrm{min}}(M)_{\sigma} \Big) + \demi \Big( e - H_{\mathrm{min}}(E)_{\omega} \Big) < k, \label{eqn:cormodmod}
\end{equation}
where we've defined $k'$ as the lower bound given in Corollary \ref{cor:unihigh}, i.e.: $k' = 9 + 2 \log (1/\varepsilon) + 1/2 \cdot \log (c+e)$.
\end{cor}
\begin{proof}
From Theorem \ref{thm:concentration}, we can ensure $\varepsilon$-locking except with probability $p$ by satisfying
\begin{equation*}
	2(CE)^2 \ln \lp \frac{40\sqrt{CE}}{\varepsilon} \sqrt{\Delta_{M,2} \Delta_{E,2}} \rp - \frac{(CE)^2}{2^8\Delta_{M, \infty} \Delta_{E, \infty}} K^2 (\varepsilon')^2 < \ln p, 
\end{equation*}
where we've defined for the time being $\varepsilon'$ as $\varepsilon - 4\Delta_{E, \infty} /\sqrt{KE}$. A quick rearrangement of the terms reveals that the inequality can be satisfied if
\begin{equation}
	\frac{2^9 \Delta_{M, \infty} \Delta_{E, \infty}}{(\varepsilon')^2} \ln \lp \frac{40 \sqrt{CE}}{\varepsilon} \sqrt{\Delta_{E,2} \Delta_{M, 2}} \lp \frac{1}{p} \rp^{1/2(CE)^2} \rp < K^2, \label{eqn:modmodproof}
\end{equation}
From our choice of $p$ we can easily see $(1/p)^{1/2(CE)^2}<2$ and from our choice of $\varepsilon$ we see that $2^9/(\varepsilon')^2 < 2^{13}/\varepsilon^2$. Thus the inequality in Equation (\ref{eqn:modmodproof}) is satisfied when
\begin{equation*}	
	13 + 2 \log \frac{1}{\varepsilon} + \log (\Delta_{M, \infty} \Delta_{E, \infty}) + \log \lp 7 + \log \frac{1}{\varepsilon} + \demi (c+e) + \demi \log (\Delta_{M, 2} \Delta_{E, 2}) \rp<2k.
\end{equation*}
However, we know that the maximum values of $\Delta_{M,2}$ and $\Delta_{E,2}$ are $M$ and $E$ respectively. Combined with our assumption that $k<c$, we can quickly reduce the above to,
\begin{equation*}
	18 + 3 \log \frac{1}{\varepsilon} + \log (c+e) + \Big( n - H_{\mathrm{min}}(M)_{\sigma} \Big) + \Big( e - H_{\mathrm{min}}(E)_{\omega} \Big)<2k.
\end{equation*}
Finally, we can identify $k'$ and give the result as desired.
\end{proof}

\section{Locking against generalized measurements}\label{sec:povms}
We will now show that the results of the previous section hold not only for projective measurements, but also for general POVMs, up to very minor changes in the various constants. The main difficulty at this point is that we cannot use Theorem \ref{thm:concentration} directly, since it only gives bounds for $(s,\eta)$-quasi-measurements. We must therefore show that a general POVM behaves essentially like an $(s,\eta)$-quasi-measurement for the purposes of the theorem. Our strategy will be probabilistic in nature: we will show that doing a general POVM $\mathcal{M}$ is mathematically equivalent to randomly selecting a measurement constructed from possible sequences of $s$ measurement results obtained from $\mathcal{M}$. With overwhelming probability, the sequence chosen will be an $(s,\eta)$-quasi-measurement, and Theorem \ref{thm:concentration} will then apply in this case.

We start by proving this last fact, namely that with very high probability, a sequence of $s$ measurement results will be an $(s,\eta)$-quasi-measurement, for an appropriately chosen $\eta$.


\begin{lem}\label{lem:locking-operator-chernoff}
	Let $\mathcal{M}^{CE' \rightarrow X}$ be any complete measurement superoperator, with $\mathcal{M}(\pi) = \sum_i \alpha_i \ketb{i}{\chi_i}\pi\ketb{\chi_i}{i}$, and consider the operator-valued random variable $Y$ which takes the value $\ketbra{\chi_i}$ with probability $\alpha_i \bra{\chi_i} \pi \ket{\chi_i} = \alpha_i/CE'$. Then, $s$ i.i.d.\ copies of $Y$ will fail to be an $(s,\eta)$-quasi-measurement with probability at most $2CE' e^{-s(\eta-1)^2/CE' 2\ln 2}$.
\end{lem}
\begin{proof}
	$Y$ fulfills all the conditions for the operator Chernoff bound (Lemma \ref{lem:operator-chernoff}) to apply, with $\mbE Y = \pi^{CE'}$. This yields
\[ \Pr \left\{ \frac{1}{s} \sum_{j=1}^s Y_j \nleqslant \eta\pi \right\} \leqslant 2CE' e^{-s(\eta-1)^2/CE'2\ln 2},  \]
and the probability on the left is an upper bound on the probability that the $s$-tuple $Y_1,\dots, Y_s$ is not an $(s,\eta)$-quasi-measurement.
\end{proof}

We now use this to show that best general POVM cannot do much better than the best $(s,\eta)$-quasi-measurement:

\begin{lem}\label{lem:general-to-nk}
	It is true that
\begin{multline}
	\sup_{\mathcal{M}} \left\| \mathcal{M}\left( \rho^{MCE'} - \rho^M \otimes \rho^{CE'} \right) \right\|_1\\
	\leqslant \max_{\mathcal{M}' \in \mathcal{L}(s,\eta)} \left\| \mathcal{M}'\left( \rho^{MCE'} - \rho^M \otimes \rho^{CE'} \right) \right\|_1 + 4(CE')^2 e^{-s(\eta-1)^2/(CE'(2\ln 2))}, 
\end{multline}
where the supremum on the left-hand side is taken over all measurement superoperators.
\end{lem}
\begin{proof}
Let $\mathcal{M}^{CE' \rightarrow X}$ be any complete measurement superoperator of the form $\mathcal{M}(\sigma) = \sum_i \alpha_i \ketb{i}{\chi_i}\sigma\ketb{\chi_i}{i}$, and define $Y$ to be the operator-valued random variable which takes value $\chi_i$ with probability $\alpha_i/CE'$. Let $Q$ be the event that $Y_1,\dots,Y_n$ is an $(s,\eta)$-quasi-measurement, where the $Y_i$ are i.i.d.\ with the same distribution as $Y$.
\begin{eqnarray*}
	\lefteqn{\lno \mcal{M} \lp \rho^{MCE'} - \rho^M \ox \rho^{CE'} \rp \rno_1} \\
	& = & \sum_i \alpha_i \lno \Tr_{CE'} \ls \chi_i \lp \rho^{MCE'} - \rho^M \ox \rho^{CE'} \rp \rs \rno_1\\
	& = & CE' \mbE_Y \lno \Tr_{CE'} \ls Y \lp \rho^{MCE'} - \rho^M \ox \rho^{CE'} \rp \rs \rno_1\\
	& = & \frac{CE'}{s} \mbE_{Y_1,\dots,Y_s} \sum_{i}^s \lno \Tr_{CE'} \ls Y_i \lp \rho^{MCE'} - \rho^M \ox \rho^{CE'} \rp \rs \rno_1.
\end{eqnarray*}
At this point we separate the expression into two terms, one for the event $Q$ and another for its complement.
\begin{eqnarray*}
	\lefteqn{\lno \mcal{M} \lp \rho^{MCE'} - \rho^M \ox \rho^{CE'} \rp \rno_1} \\
	& = & \frac{CE'}{s} \Pr\{ Q \} \mbE \left[ \left. \sum_{i=1}^s \lno \Tr_{CE'} \ls Y_i \lp \rho^{MCE'} - \rho^M \ox \rho^{CE'} \rp \rs \rno_1 \right| Q \right]\\
	& & \hspace{5mm} + \frac{CE'}{s} \Pr\{ \bar{Q} \} \mbE \left[ \left. \sum_{i=1}^s \lno \Tr_{CE'} \ls Y_i \lp \rho^{MCE'} - \rho^M \ox \rho^{CE'} \rp \rs \rno_1 \right| \bar{Q} \right]\\	
	& \leqslant &  \max_{\mathcal{M}' \in \mathcal{L}(s,\eta)} \lno \mcal{M}' \lp \rho^{MCE'} - \rho^M \ox \rho^{CE'} \rp \rno_1 \Pr\{ Q \} + 2CE' \Pr\{ \bar{Q} \} \\
	& \leqslant & \max_{\mathcal{M}' \in \mathcal{L}(s,\eta)} \lno \mcal{M}' \lp \rho^{MCE'} - \rho^M \ox \rho^{CE'} \rp \rno_1 + 4(CE')^2 e^{-s(\eta-1)^2/CE'2\ln 2}.
\end{eqnarray*}
In the above, the sum of trace distances given $Q$ was interpreted as executing an $(s,\eta)$-quasi-measurement described by $Y_1,\dots,Y_s$, and the same sum given $\bar{Q}$ was simply bounded by $2\eta$ (there are $s$ terms in the sum, each of which cannot exceed 2). In the last step, we have bounded $\Pr\{ \bar{Q} \}$ using Lemma \ref{lem:locking-operator-chernoff} and made use of the fact that we can assume without loss of generality that $|E|=|E'|$.

Finally, a non-complete measurement superoperator can always be decomposed into a complete one by splitting the POVM elements of rank greater than 1; this process always increases the trace distance.
\end{proof}

What we have achieved with the above statement is to show that the decoupling distance for a general measurement superoperator is very close to the decoupling distance of an $(s,\eta)$-quasi-measurement. All that is now left to do is to use Theorem \ref{thm:concentration} to bound the supremum over $(s,\eta)$-quasi-measurements, and we get the main theorem of this section:

\begin{thm}[Locking theorem for general measurements]\label{thm:concentrationPOVM}
Given the quantum state $\rho^{MCKEE'} = U^{CKE} \cdot (\sigma^{MCK} \ox \omega^{EE'})$ where $U$ is a random unitary operator chosen according to the Haar measure, $\sigma$ is as defined in Equation (\ref{eqn:sigmadef}) and $\omega^{EE'}$ a bipartite pure state, then
\begin{multline*}
	\PrU \lb \underset{\mcal{M}}{\mathrm{sup}} \lno \mcal{M} \lp \rho^{MCE'}\rp - \mcal{M} \lp \rho^M \ox \rho^{CE'} \rp \rno_1 >  \varepsilon \rb \\
	\leqslant \mathrm{exp} \lp 9(CE)^2\ln(CE) \ln \lp\frac{40\sqrt{CE}}{ \varepsilon} \sqrt{\Delta_{M,2}\Delta_{E,2}} \rp - \frac{(CKE)^2}{2^{10} \Delta_{M, \infty} \Delta_{E, \infty}} \lp \varepsilon - \frac{8 \Delta_{E, \infty}}{\sqrt{KE}} \rp^2 \rp.
\end{multline*}
In the above, $\Delta_{M, \infty}$, $\Delta_{M,2}$, $\Delta_{E,2}$ and $\Delta_{E, \infty}$ are as defined in Equations (\ref{eqn:deltaMmindef}), (\ref{eqn:deltaM2def}), (\ref{eqn:deltaE2def}) and (\ref{eqn:deltaEmindef}). 
\end{thm}

\begin{proof}
We may assume without loss of generality that $|E'| \leq |E|$. If not, let $E''$ be the range of $\rho^{E'} = \omega^{E'}$. Because $\omega$ is pure, $|E''| = \rank \omega^{E'} \leq |E|$. Let $V$ be the isometric embedding $E'' \hookrightarrow E'$ and $\rho^{MCE''}$ the projection of $\rho$ to $MCE''$. Then for any POVM measurement superoperator $\mcal{M}^{CE' \rightarrow X}$,
\begin{equation*}
\mcal{M}(\rho^{MCE'}) = \mcal{M}( V \rho^{MCE''} V^\dagger )
\end{equation*}
so measuring $\mcal{M}$ or $M \circ (V \cdot V^\dagger)$ will yield exactly the same measurement statistics. But the latter is a POVM on $CE''$ and $E''$ satisfies the desired dimension bound.

	Substituting the results of Lemma \ref{lem:general-to-nk} into those of Theorem \ref{thm:concentration}, we get the following:
\begin{multline}
	\PrU \lb \underset{\mcal{M}}{\mathrm{sup}} \lno \mcal{M} \lp \rho^{MCE'} - \rho^M \otimes \rho^{CE'} \rp \rno_1 \geqslant \varepsilon \rb \leqslant \mathrm{exp} \lp 2sCE \ln \lp\frac{40\sqrt{CE}}{ \varepsilon} \sqrt{\Delta_{M,2}\Delta_{E,2}} \rp  \right. \\
\left. - \frac{(CKE)^2}{2^8\eta^2 \Delta_{M, \infty} \Delta_{E, \infty}} \lp \varepsilon - 4(CE)^2 e^{-s(\eta-1)^2/(CE(2\ln 2))}  - \frac{4 \Delta_{E, \infty}}{\sqrt{KE}} \rp^2 \rp. \label{eqn:concentrationPOVMproof}
\end{multline}
We now choose $\eta=2$ and $s=(6 \ln 2) CE \ln CE$ and note that this immediately implies
\begin{equation*}
	 2(CE)^2 e^{-s(\eta-1)^2/CE2\ln 2} = \frac{2}{CE}.
\end{equation*}
We absorb this factor into our ``offset" for the $\varepsilon$ factor,
\begin{equation*}
\lp \varepsilon - 4(CE)^2 e^{-s(\eta-1)^2/CE2\ln 2} - \frac{4 \Delta_{E, \infty}}{\sqrt{KE}} \rp^2 \geqslant \lp \varepsilon - \frac{8 \Delta_{E, \infty}}{\sqrt{KE}}\rp^2.
\end{equation*} 
Substituting the choices for $s$ and $\eta$ into Equation \ref{eqn:concentrationPOVMproof} reveals the desired result.
\end{proof}

We now wish to express, in qubits, a lower bound for the key size for a given probability $p$ and a given $\varepsilon$. The relevant variables are $M=2^n$, $C=2^c$, $K=2^k$, and $E=2^e$. Unlike in the previous section, it is unnecessary to make any assumptions about the dimension of $E'$.
\begin{cor}[Locking against POVMs for a uniform message with maximal entanglement] \label{cor:unihighPOVM}
Consider the locking scheme described in Definition \ref{def:locking} for a uniform message and maximal entanglement available at the measurement. Choose $p$ and $\epsilon$ such that $\varepsilon > 16 \sqrt{1/KE}$ and $p > 2^{-9(CE)^2}$. Then the scheme will be an $\varepsilon$-locking locking scheme except with probability $p$ so long as
\begin{equation*}
	11 + 2 \log \frac{1}{\varepsilon} + \log (c+e) < k.
\end{equation*}
\end{cor}
\begin{proof}
From Theorem \ref{thm:concentrationPOVM} we can ensure $\varepsilon$-locking except with probability $p$ given
\begin{equation*}
	9 \ln(CE) \ln \lp \frac{40\sqrt{CE}}{\varepsilon} \rp + \frac{1}{9(CE')^2} \ln \frac{1}{p} < \frac{K^2(\varepsilon')^2}{2^{10}},
\end{equation*}
where we've defined for the time being $\varepsilon'$ as $\varepsilon - 8/\sqrt{KE}$. We now make use of our lower bound for $p$ as well as the assumption that $\ln(CE) \geqslant 1$ to show that the above can satisfied provided
\begin{equation*}
	9\ln(CE) \ln \lp \frac{80\sqrt{CE}}{\varepsilon} \rp < \frac{K^2 (\varepsilon')^2}{2^{10}}.
\end{equation*}
Solving the above equation for $k$ and applying the condition on $\varepsilon$ reveals that the bound can be satisfied by the statement in the lemma.
\end{proof}

\begin{cor}[Locking against POVMs for messages of bounded entropy with imperfect entanglement] \label{cor:modmodPOVM}
Consider the locking scheme described in Definition \ref{def:locking} for a uniform message and maximal entanglement available at the measurement. Choose $p$ and $\epsilon$ such that
\begin{equation*}
	\varepsilon >\frac{ 16 \Delta_{E,\infty}}{\sqrt{KE}}, \hspace{1.0in} p > 2^{-9(CE)^2}.
\end{equation*}
Then the scheme will be an $\varepsilon$-locking locking scheme except with probability $p$ so long as
\begin{equation}
	k' + \demi \Big( n - H_{\mathrm{min}}(M)_{\sigma} \Big) + \demi \Big( e - H_{\mathrm{min}}(E)_{\omega} \Big) < k, \label{eqn:cormodmodPOVM}
\end{equation}
where we've defined $k'$ as the lower bound given in Corollary \ref{cor:unihighPOVM}, i.e.: $k' = 11 + 2 \log (1/\varepsilon) + \log (c+e)$.
\end{cor}
\begin{proof}
From Theorem \ref{thm:concentration}, we can ensure $\varepsilon$-locking with probability $p$ by satisfying,
From Theorem \ref{thm:concentrationPOVM} we can ensure $\varepsilon$-locking with probability $p$ given
\begin{equation*}
	9 \ln(CE) \ln \lp \frac{40\sqrt{CE}\sqrt{\Delta_{M,2} \Delta_{E,2}}}{\varepsilon} \rp + \frac{1}{9(CE')^2} \ln \frac{1}{p} < \frac{K^2(\varepsilon')^2}{2^{10}\Delta_{M, \infty} \Delta_{E, \infty}},
\end{equation*}
where we've defined for the time being $\varepsilon'$ as $\varepsilon - 8/\sqrt{KE}$.  We now make use of our lower bound for $p$ as well as the assumption that $\ln(CE) \geqslant 1$ to show that the above can satisfied provided
\begin{equation*}
	9\ln(CE) \ln \lp \frac{80\sqrt{CE}\sqrt{\Delta_{M,2} \Delta_{E,2}}}{\varepsilon} \rp < \frac{K^2 (\varepsilon')^2}{2^{10}\Delta_{M, \infty} \Delta_{E, \infty}} .
\end{equation*}
Next, we use our definition for $\varepsilon'$ and our bound for $\varepsilon$ and we solve for $k$ to find that the bound is satisfied provided
\begin{equation*}
	21 + 3 \log \frac{1}{\varepsilon} + 2 \log(c+e) + \log (\Delta_{M, \infty} \Delta_{E, \infty}) < 2k .
\end{equation*}
Finally, we can identify $k'$ and give the result as desired.
\end{proof}

The lower bound requirement on $\varepsilon$ in Corollary~\ref{cor:modmodPOVM} limits the corollary's range of applicability to situations in which $H_{\mathrm{min}}(E)_\omega$ is not too small. Specifically, the requirement can be rewritten in light of (\ref{eqn:cormodmodPOVM}) as
\begin{equation*}
2 \log (c+e) + \left( n - H_{\mathrm{min}}(M)_\sigma \right) + 3 H_{\mathrm{min}}(E)_\omega
>
e + \mathrm{const}.
\end{equation*}
So, at least when the message is uniform, the requirement is roughly that $H_{\mathrm{min}}(E)_\sigma > e/3$. We suspect that this requirement can be eliminated but leave it as an open problem to find a way to do so.

\section{Locking versus decodability}\label{sec:packing}
The previous sections have shown that, under certain conditions, no classical information is recoverable by the receiver. Here we aim to show that, in many regimes, these results are essentially optimal. We do this by showing that if we make the key only very slightly smaller, then with overwhelming probability, the classical message will be decodable with a negligible error probability. In fact we prove even more: in this regime where the information is decodable, the decoder can even decode a \emph{purification} of the classical message. In other words, in this generic scenario where $U$ is chosen with no preferred basis, either all \emph{classical} information is locked away, or we can decode \emph{quantum} information.  This is formalized in the next theorem.

In order to study decodability, we must discard the identifications made in Figure \ref{fig:locking} to study locking and return to the original scenario described by Figure \ref{fig:scenario}. Whereas $k$ was previously the number of qubits in system $K$, there is no system $K$ in Figure \ref{fig:locking}. Instead, we define $k=n-c$, which is consistent with its earlier definition. Now, however, it might be the case that $k$ is negative since decoding could require the cyphertext to be longer than the message.

The following theorem generalizes the discussion of Section~\ref{sec:how-to-lock} to nonuniform messages and imperfect entanglement.
\begin{thm} \label{thm:decode}
	If $U$ is chosen according to the Haar measure, then the information in the scheme described in Figure \ref{fig:scenario} is such that there exists a decoding CPTP map $\mathcal{D}^{CE' \rightarrow N}$ such that
\begin{equation*}
	\left\| \mathcal{D}\left( \tr_{D} \left[U^{NE \rightarrow CD}\left( \sigma^{RMN} \otimes \omega^{E'E} \right) (U^{NE \rightarrow CD})\mdag\right] \right) - \sigma^{RMN}\right\|_1 \leqslant \varepsilon
\end{equation*}
asymptotically almost surely, where $\sigma^{RMN}$ is a purification of $\sigma^{MN}$, as long as
\[ k \leqslant \demi \Big( n - H_{\max}(M)_{\sigma} \Big) - \demi \Big( e - H_2(E)_{\omega} \Big) - 2 \log(1/\varepsilon) - 4 \]
\end{thm}
\begin{proof}
	Using Theorem 3.7 from \cite{fred-these}, we get that
\begin{equation*}
	\mbE_U \left\| \tr_{C}\left[ U^{NE \rightarrow CD}\left( \sigma^{RMN} \otimes \omega^E \right) (U^{NE \rightarrow CD})\mdag \right] - \sigma^{RM} \otimes \rho^{D} \right\|_1 \leqslant 2^{\demi H_{\max}(M)_{\sigma} - \demi H_{2}(E)_{\omega}} \sqrt{\frac{D}{C}}.
\end{equation*}
It can also be shown that the value of this trace distance will asymptotically almost surely not exceed twice this bound. Under this condition, we have that:
\begin{equation*}
	\left\| \tr_{C}\left[ U^{NE \rightarrow CD}\left( \sigma^{RMN} \otimes \omega^E \right) (U^{NE \rightarrow CD})\mdag \right] - \sigma^{RM} \otimes \rho^{D} \right\|_1 \leqslant 2 \times 2^{\demi H_{\max}(M)_{\sigma} - \demi H_2(E)_{\omega}} \sqrt{\frac{D}{C}}.
\end{equation*}
Uhlmann's Theorem then implies the existence of a partial isometry $V^{CE' \rightarrow NG}$ and of a purification of $\rho^{D}$ on system $G$ that we call $\theta^{DG}$ such that
\begin{equation*}
	\left\| VU \left( \sigma^{RMN} \otimes \omega^{E'E} \right) U\mdag V\mdag - \sigma^{RMN}  \otimes \theta^{DG} \right\|_1 \leqslant 4 \left( 2^{H_{\max}(M)_{\sigma} - H_2(E)_{\omega}} \frac{D}{C} \right)^{1/4}.
\end{equation*}
Defining $\mathcal{D}^{CE' \rightarrow N}$ as $\mathcal{D}(\xi) = \tr_{G}\left[ V \xi V\mdag \right]$ and tracing out system $D$, we get that
\begin{equation*}
	\left\| \mathcal{D}\left( \tr_{D}\left[ U^{NE \rightarrow CD}\left( \sigma^{RMN} \otimes \omega^{E'E} \right) (U^{NE \rightarrow CD})\mdag \right] \right) - \sigma^{RMN}\right\|_1 \leqslant 4 \left( 2^{H_{\max}(M)_{\sigma} - H_2(E)_{\omega}} \frac{D}{C} \right)^{1/4}.
\end{equation*}

Now, to satisfy the theorem statement, we need to ensure that
\[ 4 \left( 2^{H_{\max}(M)_{\sigma} - H_2(E)_{\omega}} \frac{D}{C} \right)^{1/4} \leqslant \varepsilon. \]
Taking logarithms on both sides and using the fact that $\log D = k + e$, we get that

\begin{equation*}
	2 + \frac{1}{4} \left[ H_{\max}(M)_{\sigma} - H_2(E)_{\omega} + e + k - c \right] \leqslant \log \varepsilon.
\end{equation*}
Substituting in the fact that $c = n - k$, we arrive at the statement of the theorem.
\end{proof}

\section{Implications for the security of quantum protocols against quantum adversaries}\label{sec:qkd}
When designing quantum cryptographic protocols, it is often necessary to show that a quantum adversary (``Eve'') is left with only a negligible amount of information on some secret string. An initial attempt at formalizing this idea is to say that, at the end of the protocol, regardless of what measurement Eve makes on her quantum system, the mutual information between her measurement result and the secret string is at most $\varepsilon$ (in other words, her accessible information about the message is at most $\varepsilon$). This was often taken as the security definition for quantum key distribution, usually implicitly by simply not considering that the adversary might keep quantum data at the end of the protocol \cite{lc99,sp00,NC2000,gl03,lca05} (see also discussion in \cite{bhlmo05,RK04,krbm07}). In \cite{krbm07}, it is shown that this definition of security is inadequate, precisely because of possible locking effects. Indeed, this security definition does not exclude the possibility that Eve, upon gaining partial knowledge of $S$ after the end of the protocol, could then gain more by making a measurement on her quantum register that depends on the partial information that she has learned. In \cite{krbm07}, the authors exhibit an admittedly contrived quantum key distribution protocol which generates a secret $n$-bit key such that, if Eve learns the first $n-1$ bits, she can then learn the remaining bit by measuring her own quantum register. 

The locking scheme presented above allows us to demonstrate a much more spectacular failure of this security definition. We will show that there exists a quantum key distribution protocol that ensures that an adversary has negligible accessible information about the final key, but with which an adversary can recover the entire key upon learning only a very small fraction of it.

\subsection{Description of the protocol}
We will derive this faulty protocol by starting with a protocol that is truly secure, and then making Alice send a locked version of the secret string directly to Eve. We will be able to prove that regardless of what measurement Eve makes on her state, she will learn essentially no information on the string, but of course, she only needs to learn a tiny amount of information to unlock what Alice sent her. More precisely, let $P$ be a quantum key distribution protocol such that, at the end of its execution, Alice and Bob share an $n$-bit string, and Eve has a quantum state representing everything that she has managed to learn about the string. We will also assume that $P$ is a truly secure protocol: the string together with Eve's quantum state can be represented as a quantum state $\sigma^{SE}$ such that $\| \sigma^{SE} - \pi^S \otimes \sigma^E \|_1 \leqslant \varepsilon$, where $S$ is a quantum register holding the secret string, and $E$ is Eve's quantum register. Now, we will define the protocol $P'$ to be the following quantum key distribution protocol: Alice and Bob first run $P$ to generate a string $s$ of length $n$, and then Alice splits $s$ into two parts: the first part $s_k$ is of size $O(\log n)$, and the second part $s_c$ contains the rest of the key. Alice then uses the classical key $s_k$ to create a quantum state in register $C$ that contains a locked version of $s_c$ and sends the system $C$ to Eve.

How secure is $P'$? It is clearly very insecure, since, if Eve ever ends up learning $s_k$ (via a known plaintext attack, for instance), she can then completely recover $s_c$. However, the next theorem shows that, right after the execution of $P'$, Eve cannot make any measurement that will reveal information about the key. In particular, $P'$ satisfies the requirement that Eve's accessible information on the key be very low.

\begin{thm}
Let $P$ and $P'$ be quantum key distribution protocols as defined as above, and let $\rho^{CES}$ be the state at the end of the execution of $P'$: $S$ contains the $n$-bit string $s$, $E$ is Eve's quantum register after the execution of $P$, and $C$ contains the locked version of $s_c$ that Alice sent to Eve. Then, for any measurement superoperator $\mathcal{M}^{CE \rightarrow X}$, there exists a state $\xi^X$ such that
\[ \left\| \mathcal{M}(\rho^{CES}) - \xi^X \otimes \pi^S \right\|_1 \leqslant 2\varepsilon. \]
This also entails that
\[ I_{\acc}(S;CE) \leqslant 8\varepsilon n + 2\eta(1-2\varepsilon) + 2\eta(2\varepsilon) \]
via the Alicki-Fannes inequality (see Lemma \ref{lem:alicki-fannes}).
\end{thm}
\begin{proof}
From the definition of $P$, we have that
\begin{equation}
\left\| \rho^{ES} - \pi^S \otimes \rho^E \right\|_1 \leqslant \varepsilon.
\end{equation}
Now, let $\mathcal{C}^{S \rightarrow CS}$ be a superoperator that takes a classical string in $S$, splits it into $s_k$ and $s_c$, creates a locked version of $s_c$ with $s_k$ as the key into the quantum system $C$, and leaves the classical string in $S$ unchanged; this is simply the operation that Alice performs when preparing $C$ for Eve. The above inequality, combined with the monotonicity of the trace distance under CPTP maps yields
\begin{equation}
\left\| \rho^{CES} -  \mathcal{C}(\pi^S) \otimes \rho^E \right\|_1 \leqslant \varepsilon
\label{eqn:qkd1}
\end{equation}
and hence, for any measurement superoperator $\mathcal{M}^{CE \rightarrow X}$,
\begin{equation}
\left\| \mathcal{M}(\rho^{CES}) -  \mathcal{M}(\mathcal{C}(\pi^S) \otimes \rho^E) \right\|_1 \leqslant \varepsilon
\label{eqn:qkd2}
\end{equation}
Consider now the expression $\mathcal{M}^{CE \rightarrow X}(\mathcal{C}(\pi^S) \otimes \rho^E)$: it can be viewed as a measurement on the $C$ system of $\mathcal{C}^{S \rightarrow CS}(\pi^S)$ alone that is implemented by creating the state $\rho^E$ and then measuring $\mathcal{M}^{CE \rightarrow X}$. Furthermore, note that, by the definition of an $\varepsilon$-locking scheme, we have that, for every measurement superoperator $\mathcal{N}^{C \rightarrow X}$,
\begin{equation}
\left\| \mathcal{N}(\mathcal{C}(\pi^S)) - \mathcal{N}(\tr_S[\mathcal{C}(\pi^S)]) \otimes \pi^S \right\|_1 \leqslant \varepsilon.
\end{equation}
Applying this to $\mathcal{M}^{CE \rightarrow X}(\mathcal{C}(\pi^S) \otimes \rho^E)$, we get that
\begin{equation}
\left\| \mathcal{M}(\mathcal{C}(\pi^S) \otimes \rho^E) -  \mathcal{M}(\tr_S[\mathcal{C}(\pi^S)] \otimes \rho^E) \otimes \pi^S \right\|_1 \leqslant \varepsilon.
\label{eqn:qkd3}
\end{equation}
We now use the triangle inequality on Equations (\ref{eqn:qkd2}) and (\ref{eqn:qkd3}) to obtain
\begin{equation}
\left\| \mathcal{M}(\rho^{CES}) - \mathcal{M}(\tr_S[\mathcal{C}(\pi^S)] \otimes \rho^E) \otimes \pi^S \right\|_1 \leqslant 2\varepsilon
\end{equation}
which yields the theorem with $\xi^X := \mathcal{M}(\tr_S[\mathcal{C}(\pi^S)] \otimes \rho^E)$.
\end{proof}

Hence, we have shown that requiring that Eve's accessible information on the generated key be low is not an adequate definition of security for quantum key distribution. We have exhibited a protocol $P'$ which guarantees low accessible information and yet is clearly insecure due to locking effects.

\section{Discussion}\label{sec:discussion}

It is natural in physics to  measure the ``correlation'' between two quantum physical systems using the correlation between the outcomes of measurements on those two systems. Two-point correlation functions are but the most ubiquitous examples. The results in this article demonstrate that this practice can sometimes be very misleading. The $\epsilon$-locking quantum states exhibited in this article would reveal no correlations using any type of measurement, but enlarging one of the two systems by a small number of qubits would expose near-perfect correlation between the two systems. This is an important and counterintuitive property of information in quantum mechanical systems: measurements can be distressingly bad ways to detect correlation.

The extensive literature on quantum discord is essentially devoted to exploring the relationship between accessible, or classical, and quantum mutual information~\cite{OlZur01,HenVed01,BluZur06}. Since the discord is defined as the gap between the quantum and classical mutual informations, locking can be viewed as the extreme case where classical mutual information doesn't detect any of the very abundant quantum mutual information.
Previous work had demonstrated that transmitting a constant number of physical qubits can cause the classical mutual information to increase from a fixed small constant to an arbitrarily large value. In this article, we have strengthened the definition of locking, replacing the mutual information by the trace distance to a product distribution. Moreover, we have shown that the locking effect still exists even when the trace distance (or the classical mutual information) is made arbitrarily small. In light of these results, claims that the discord is a robust measure of quantum correlation~\cite{WerSou09} should treated with skepticism. While discord is certainly a signature of quantumness, its susceptibility to locking means that it is in this important respect not robust.

Previous studies of information locking had also always focused on the example of sending classical information in one of a small number of different bases unknown to the receiver. The intuition was that a receiver ignorant of the basis could not do much better than guessing the basis and then measuring. Most of the time, he would guess incorrectly and his measurement would then destroy the information. Moving away from that paradigm, in this article we consider classical information encoded using a single generic unitary transformation mixing the input information with half of an entangled state shared with the receiver. The ``key'' then becomes a quantum system. While the original paradigm can be recovered by eliminating the entanglement and encrypting the key quantum system with a private quantum channel, the setting considered here is strictly more general. 

Indeed, we find that, for an $n$-bit uniform message and maximal entanglement, the information is generically $\epsilon$-locked until the receiver is within $O(\log n/\epsilon)$ qubits of being able to completely decode the message. Our definition of locking is stronger than those previously studied and our results imply, for the first time, that the classical mutual information can be made arbitrarily small. Our method of proof in the case of projective measurements was a fairly standard discretization argument but the extension to POVM measurements required a new strategy exploiting the operator Chernoff bound. In contrast to previous studies of locking, we do not require the message to be uniformly distributed, working instead with a min-entropy bound on the distribution of messages. In that case, we found that the key size was at most the gap between the max- and min-entropies of the message, modulo the logarithmic terms that dominate in the uniform situation. 

For information theorists, this may appear reminiscent of a strong converse to a channel capacity problem. Roughly, a strong converse theorem states that any attempt to transmit above the channel capacity will result in the decoding error probability approaching one. In our setting, the analog of the strong converse would be a matching lower bound to Equation (\ref{eqn:measurement-existence}) of the form
\begin{equation} \label{eqn:strong-converse}
1 - \epsilon < \frac{1}{M} \sum_m \sum_{m' \neq m} p(m' | m )
\end{equation}
whenever $C < M$, indicating the the probability of incorrectly decoding the message is at least $1-\epsilon$. What we prove here is {much} stronger. Equation (\ref{eqn:strong-converse}) doesn't rule out the possibility of being able to pin the message down to some relatively small set. More generally, it doesn't imply a small mutual information between the message and the measurement outcome. Information locking does imply these stronger statements.

As such, information locking has a natural cryptographic interpretation even if we haven't emphasized it in this article. The special case of our scenario mentioned above, with no entanglement and a quantum key encrypted using a private quantum channel, leads to a method for encrypting classical messages using a secret key of size independent of the length of the message. Similarly, information locking schemes can be used to construct string commitment protocols with surprisingly good parameters~\cite{BuhChrHay06,BuhChrHay08}. These cryptographic applications are emphasized in the companion article~\cite{FawHaySen10}.

To the extent that random unitary transformations provide good models of black hole evaporation, our results might also have implications for that process. Oppenheim and Smolin had previously suggested that information locking could rescue the long-lived remnant hypothesis~\cite{OppSmo06}. In essence, their idea was that a remnant with a small number of states could lock all the information of a large black hole, thereby evading the inconsistencies with low energy physics that arise from having large  numbers of remnant species~\cite{AhaCasNus,CarWil}. Their proposal, however, relied on previously studied locking states that treated the encoded message and the key very differently. Consequently, the proposal required that the black hole keep hold of the key until the very last moments of its evaporation, implying some ad hoc dynamical distinction between encoded message and key in the evaporation process. Our results imply that \emph{if} the dynamics are well-modeled by a Haar random unitary transformation, then any small portion of the output system can be used as the key. No ad hoc distinction is necessary. 

Ironically, the information locking effect is also perfectly compatible with the rapid release of information from a black hole predicted in~\cite{HayPre07}, assuming a unitary evaporation process. That article observed that if a black hole is already highly entangled with Hawking radiation from an earlier time, then messages would be released from the black hole in the Hawking radiation once the black hole dynamics had sufficiently ``scrambled'' the message with internal black hole degrees of freedom. By virtue of the fact that we treat generic unitary transformations acting on a message and half of an entangled state, our results apply to the setting of that paper and the followup~\cite{SusSek}. Specifically, our results imply that in the case of a larger message, \emph{no} information about the message could be obtained from the Hawking radiation until moments before it could \emph{all} be obtained. The conclusion depends, of course, on whether the random unitary transformation is a good model of the evaporation process. While the generic unitary transformations considered here would take exponential time to implement on a quantum computer, the follow-up article~\cite{FawHaySen10} shows, at least, that locking can be achieved with a quantum circuit of depth only slightly superlinear in the number of qubits in the system. Other attempts to apply random unitary transformations to the black hole information problem, such as~\cite{Lloyd06,BraunSomZ09}, will be affected similarly by information locking.

To summarize, this article defined information locking more stringently than previously and nonetheless found that this stronger form of locking is generic: if information is encoded using a random unitary transformation, then it will either be decodable or locked. Almost no middle ground occurs. This observation has implications for cryptography and, potentially, for black hole physics. 

\section*{Acknowledgments}

Andreas Winter has independently established some locking results for generic unitary transformations. We would like to thank Jonathan Oppenheim for helpful discussions and the Mittag-Leffler Institute for its kind hospitality.  This research was supported by the Canada Research Chairs program, the Perimeter Institute, CIFAR, CFI, FQRNT's INTRIQ, MITACS, NSERC, ORF, ONR through grant N000140811249, QuantumWorks, and the Swiss National Science Foundation through grant no. 200021-119868.

\bibliographystyle{alpha}
\bibliography{big}

\newpage
\appendix
\section{Miscellany}\label{app}
\begin{defin}[Lipschitz constant]\label{def:lipschitz}
	Let $f: \mfX \rightarrow \mfY$ be a function from the metric space $(\mfX, d_{\mfX})$ to the metric space $(\mfY, d_{\mfY})$. Then, the Lipschitz constant of $f$ is defined as
\[ \sup_{x_1, x_2 \in \mfX} \frac{d_{\mfY}(f(x_1), f(x_2))}{d_{\mfX}(x_1,x_2)}. \]
If the above quantity is not bounded, the constant is not defined.
\end{defin}

\begin{lem}[Lemma IV.3 in~\cite{FQSW}]\label{Schur}
	For any matrix $X^{A\overline{A}R}$ and for $\mathrm{d}U$ the Haar measure over unitaries, we have the following property:
\begin{equation*}
	\int_U \lp U_A \ox U_{\overline{A}} \otimes \ident^R \rp X^{A\overline{A}R} \lp U_A^{\dagger} \ox U_{\overline{A}}^{\dagger} \otimes \ident^R \rp \mathrm{d}U = \alpha_{+} \lp X \rp \otimes \Pi_+^A + \alpha_{-} \lp X \rp \otimes \Pi_-^A
\end{equation*}
where
\begin{equation*}
	\alpha_{\pm}\lp X \rp = \frac{\Tr_{A\overline{A}}\left[ X (\Pi^A_{\pm} \otimes \ident^R) \right] }{\mathrm{rank} \lp \Pi^A_{\pm} \rp} \hspace{0.6in} \Pi^A_{\pm} = \frac{1}{2} \lp \ident^{A\overline{A}} \pm F^{A}_{\overline{A}} \rp \hspace{0.6in} \mathrm{rank} \lp \Pi^A_{\pm} \rp = \frac{|A| (|A| \pm 1)}{2}.
\end{equation*}
\end{lem}
\bigskip

\begin{lem}[Operator Chernoff bound~ \cite{ahlswede-winter}]\label{lem:operator-chernoff}
  Let $X_1,\ldots,X_M$ be i.i.d.\ random variables taking values in the operators $\Pos(\sfA)$, with $0\leqslant X_j\leqslant \ident$, with $A=\mbE X_j\geqslant\alpha \ident$, and let $0<\eta \leqslant 1/2$. Then
  \begin{equation}
    \Pr \left\{ \frac{1}{M}\sum_{j=1}^M X_j \nleqslant (1+\eta)A \right\} \leqslant 2|A| \exp\left( -M\frac{\alpha\eta^2}{2\ln 2} \right).
  \end{equation}
\end{lem}

\begin{lem}[Trace distance versus Euclidean norm for pure states (See, e.g. ~\cite{NC2000}.)] \label{1normto2norm}
Consider any two quantum states $\ket{\varphi}, \ket{\tilde{\varphi}}$ with density associated operators  $\varphi, \tilde{\varphi}$ respectively. We can relate the $1$-norm distance between the operators to the $2$-norm distance of the states as follows,
\begin{equation*}
	\lno \varphi -\tilde{\varphi} \rno_1 \leq 2 \lno \ket{\varphi} - \ket{\tilde{\varphi}} \rno_2.
\end{equation*}
\end{lem}

\begin{lem}[A bound for the $1$-norm in terms of conditional entropy~\cite{renner-phd,fred-these}] \label{lem:tight12}
Let $\rho \in \mathrm{L}(A)$ be any Hermitian operator and let $\gamma \in \mathrm{Pos}(A)$ be a positive definite operator. Then,
\begin{equation*}
	\lno \rho \rno_1 \leqslant \sqrt{\Tr \ls \gamma \rs \Tr \ls \lp \gamma^{-1/4} \rho \gamma^{-1/4} \rp^2 \rs}.
\end{equation*}
\end{lem}
\begin{proof}
\begin{eqnarray*}
	\lno \rho \rno_1 & = & \underset{U \in \mcal{U}(A)}{\mathrm{max}} \left| \Tr \ls U \rho \rs \right| \\
	& = & \underset{U \in \mcal{U}(A)}{\mathrm{max}} \left| \Tr \ls \lp \gamma^{1/4} U \gamma^{1/4} \rp \lp \gamma^{-1/4} \: \rho \: \gamma^{-1/4} \rp \rs \right| \\
	& \leqslant & \underset{U \in \mcal{U}(A)}{\mathrm{max}} \sqrt{ \Tr \ls \lp \gamma^{1/4} U \gamma^{1/4} \rp \lp \gamma^{1/4} U^{\dagger} \gamma^{1/4} \rp \rs \Tr \ls \gamma^{-1/4} \: \rho \: \gamma^{-1/2} \: \rho^{\dagger} \: \gamma^{-1/4} \rs } \\
	& = & \sqrt{\underset{U \in \mcal{U}(A)}{\mathrm{max}} \Tr \ls \gamma^{1/2} U \gamma^{1/2} U^{\dagger} \rs \Tr \ls \gamma^{-1/4} \: \rho \: \gamma^{-1/2} \: \rho^{\dagger} \: \gamma^{-1/4} \rs } \\
	& = & \sqrt{\Tr \ls \gamma \rs \Tr \ls \gamma^{-1/4} \: \rho \: \gamma^{-1/2} \: \rho^{\dagger} \: \gamma^{-1/4} \rs } ,
\end{eqnarray*}
where the first equality is an application of Lemma I.6 in~\cite{fred-these} and the inequality results from an application of Cauchy-Schwarz, and the maximizations are over all unitaries on $A$. The last equality follows from
\begin{eqnarray*}
	\underset{U \in \mcal{U}(A)}{\mathrm{max}} \Tr \ls \gamma^{1/2} U \gamma^{1/2} U^{\dagger} \rs & \leqslant & \underset{U \in \mcal{U}(A)}{\mathrm{max}} \sqrt{\Tr \ls \gamma \rs \Tr \ls U \gamma^{1/2} U^{\dagger} U \gamma ^{1/2} U^{\dagger} \rs} \\
	& = & \Tr \ls \gamma \rs \\
	& \leqslant & \underset{U \in \mcal{U}(A)}{\mathrm{max}} \Tr \ls \gamma^{1/2} U \gamma^{1/2} U^{\dagger} \rs.
\end{eqnarray*}
\end{proof}

\end{document}